\documentclass[11pt,a4paper]{article}
\usepackage{jheppub} 
\usepackage{amsthm} 
\usepackage[all]{xy} 
\usepackage{marvosym} 
 
\begin{document} 
 
\numberwithin{equation}{section} 
\newtheorem{example}{Example}[section] 
\newtheorem{theorem}{Theorem}[section] 
\newtheorem{corollary}{Corollary}[section] 
\newtheorem{rem}{Remark}[section] 
\newtheorem{defi}{Definition}[section] 
\newtheorem{conj}{Conjecture}[section]

\newenvironment{remark}[1][Remark]{\begin{trivlist} 
\item[\hskip \labelsep {\bfseries #1}]}{\end{trivlist}} 

\newcommand{\cO}{\mathcal{O}} 
\newcommand{\PP}{\mathbb{P}} 
\newcommand{\Fuk}{\text{\sf Fuk}} 
\newcommand{\C}{\mathbb{C}} 
\newcommand{\FS}{\text{\sf FS}} 
\newcommand{\LG}{\text{\sf LG}} 
\newcommand{\tritime}{\text{\Clocklogo}} 
 
\title{Super Landau-Ginzburg mirrors and algebraic cycles} 
 
\author[a]{Richard S. Garavuso,}  
 
\author[b,c]{Ludmil Katzarkov,} 
 
\author[d]{Maximilian Kreuzer} 
 
\author[b]{and Alexander Noll} 
 
\affiliation[a]{Department of Mathematical and Statistical Sciences,  
                University of Alberta  
\\ 
                632 Central Academic Building;  
                Edmonton, Alberta T6G 2G1; Canada} 
 
\affiliation[b]{Department of Mathematics, Universit\"{a}t Wien  
\\ 
                Garnisongasse 3, Vienna A-1090; Austria} 
 
\affiliation[c]{Department of Mathematics, University of Miami 
\\ 
                PO Box 249085; Coral Gables, FL 33124-4250; United States}  
 
\affiliation[d]{Institut f\"{u}r Theoretische Physik,  
                Technische Universit\"{a}t Wien 
\\ 
                Wiedner Hauptstrasse 8-10/136; Vienna A-1040; Austria} 
 
\emailAdd{garavuso@ualberta.ca}  
\emailAdd{lkatzark@math.uci.edu} 
\emailAdd{nolla5@univie.ac.at} 
 
\abstract{We investigate the super Landau-Ginzburg mirrors of gauged 
linear sigma models which, in an appropriate low energy limit, reduce 
to nonlinear sigma models with K\"{a}hler supermanifold target spaces of 
nonnegative super-first Chern class.} 
 
\keywords{Conformal Field Models in String Theory, Sigma Models, String Duality} 
 
\arxivnumber{} 
 
\dedicated{Dedicated to our dear friend and coauthor Maximilian 
Kreuzer, who died November 26, 2010.}

\maketitle 
 
\flushbottom 
 
\section{Introduction} 
 
A K\"{a}hler manifold $ M $ with nonnegative first Chern class can be  
described in terms of a $ (2,2) $ supersymmetric gauged linear sigma model in  
$ 1 + 1 $ dimensions \cite{Witten:Phases,MorrisonPlesser:Summing}. 
That is, in an appropriate low energy limit, the gauge theory reduces to a  
nonlinear sigma model with resolved target space $ M $. 
Hori and Vafa \cite{HoriVafa:Mirror} proved that the gauged linear sigma model  
corresponding to $ M $ is mirror to a Landau-Ginzburg theory. 
When $ M $ is a Calabi-Yau manifold, the mirror Landau-Ginzburg  
theory can sometimes be given a geometrical interpretation in terms of a  
nonlinear sigma model with Calabi-Yau resolved target space $ \widetilde{M} $. 
Here, the manifold $ \widetilde{M} $ is the mirror of $ M $. 
 
A rigid Calabi-Yau manifold has no complex structure moduli. 
The mirror of such a manifold has no K\"{a}hler moduli and hence cannot be a  
K\"{a}hler manifold in the conventional sense. 
Thus, K\"{a}hler manifolds cannot be the most general geometrical framework 
for understanding mirror symmetry. 
The first progress towards generalizing this framework came when it was  
suggested that higher-dimensional Fano varieties could provide the mirrors for 
rigid Calabi-Yau manifolds  
\cite{Schimmrigk,CandelasDerrickParkes:Generalized}. 
Later progress came when Sethi \cite{Sethi} proposed a general correspondence  
between orbifolds of (2,2) supersymmetric Landau-Ginzburg theories in  
$ 1 + 1 $ dimensions with integral $ \hat{c} \equiv c/3 $ (where $ c $ is  
the central charge) and nonlinear sigma models.     
Here, the resolved target space of the nonlinear sigma model is either a Calabi-Yau  
manifold or a Calabi-Yau \emph{super}manifold. 
Using this proposal, Sethi argued that the mirror of a rigid Calabi-Yau 
manifold is a Calabi-Yau supermanifold and that supervarieties are the proper  
geometrical framework of mirror symmetry. 
 
In the absence of a proper supercohomology theory, Sethi used heuristic  
arguments for computing the Hodge numbers of Calabi-Yau supermanifolds 
realized (at least in principle) as crepant resolutions of hypersurfaces in  
weighted complex superprojective spaces. 
Using these same heuristic arguments, the geometrical interpretations  
associated with Sethi's proposed correspondence were investigated in  
\cite{GaravusoKreuzerNoll:Fano}. 
As might be expected, it was found that the supermanifold Hodge numbers  
obtained by using these heuristic arguments do not always agree with those of 
the corresponding Landau-Ginzburg orbifold. 
 
Aganagic and Vafa \cite{AganagicVafa:Mirror} showed that when $ M $ is a  
Calabi-Yau supermanifold realized as a crepant resolution of a  
weighted complex superprojective space, the corresponding gauged linear  
sigma model is mirror to a \emph{super} Landau-Ginzburg theory. 
This relation should also hold when $ M $ is any K\"{a}hler supermanifold 
with nonnegative super-first Chern class. 
In the super Calabi-Yau case, one expects to obtain a geometrical intepretation 
of the super Landau-Ginzburg mirror, i.e. a nonlinear sigma model with super 
Calabi-Yau resolved target space $ \widetilde{M} $. 
Various examples of such geometrical interpretations are given in  
\cite{AganagicVafa:Mirror,Ahn:Mirror,BelhajDrissiRasmussenSaidiSebbar:Toric, 
Ricci:Super,AhlBelhajDrissiSaidi:On}. 
 
In this paper, we will discuss the super Landau-Ginzburg mirrors of gauged  
linear sigma models corresponding to K\"{a}hler supermanifolds with nonnegative 
super-first Chern class. 
In Sections \ref{WCP}, \ref{Tot}, \ref{WCP[s]}, \ref{WCP[s_1,...,s_l]}, and  
\ref{T[s_A1,...,s_Al]}, respectively, we will consider the cases in which the  
K\"{a}hler  supermanifold is realized as a crepant resolution of 
\begin{enumerate} 
 
\item[(i)]  
a weighted complex superprojective space $ \mathbf{WCP}^{m-1|n} $, 
\item[(ii)] 
$ \mathrm{Tot}  
  \left(  
         \mathcal{O}(-s) \rightarrow\mathbf{WCP}^{m-1|n}  
  \right) $, 
 
\item[(iii)] 
a hypersurface in $ \mathbf{WCP}^{m-1|n} $, 
 
\item[(iv)] 
a complete intersection in $ \mathbf{WCP}^{m-1|n} $, and 
 
\item[(v)] 
a complete intersection in a general toric supervariety. 
 
\end{enumerate}
For cases (i) - (iv), the unresolved variety corresponds to a gauged linear sigma model with 
$ U(1) $ gauge group. 
For case (iii), when the super Calabi-Yau condition is satisfied, we will give 
a geometrical interpretation of the super Landau-Ginzburg mirror. 
Our analysis here will include two examples which can be used to test Sethi's proposed 
correspondence using the recently developed techniques of \cite{GaravusoKatzarkovNoll}. 
For case (v), we will show that the mirror theory can be expressed as a super  
Landau-Ginzburg theory on a noncompact supermanifold. 
In Section \ref{PeriodRelations}, we will establish some relations between  
periods of mirrors of gauged linear sigma models corresponding to various  
geometries. 
These relations generalize results which were obtained in  
\cite{Schwarz:Sigma} by working with A-models. \pagebreak
In Section \ref{Categorical}, we will give a categorical interpretation of the 
material presented in Sections \ref{WCP} - \ref{PeriodRelations}. 
Finally, in the Appendix, we will prove a theorem and corollary concerning the 
quasihomogeneity of the mirror superpotential for case (iii).

\section{\label{WCP}\texorpdfstring{$\mathbf{WCP}^{m-1|n}$}{WCPm-1|n}} 
 
Consider a weighted complex superprojective space 
$ \mathbf{WCP}^{m-1|n}_{(Q_1,\ldots,Q_m|q_1,\ldots,q_n)} $ 
with $ m $ homogeneous bosonic coordinates $ \phi_i $ having weights  
$ Q_i $, where $ i=1,\ldots,m $, and $ n $ homogeneous fermionic coordinates  
$ \xi_a $ having weights $ q_a $, where $ a = 1,\ldots,n $, i.e. 
\begin{equation*} 
\left(  
       \phi_1,\ldots,\phi_m | \xi_1,\ldots,\xi_n 
\right) 
\simeq 
\left( 
       \lambda^{Q_1}\phi_1,\ldots,\lambda^{Q_m}\phi_m |  
       \lambda^{q_1}\xi_1,\ldots,\lambda^{q_n}\xi_n 
\right), 
\quad 
\lambda \in \mathbf{C}^{*} \, .  
\end{equation*} 
A K\"{a}hler supermanifold may be obtained as a crepant resolution of this  
supervariety if a crepant resolution exists. 
Such a supermanifold would have nonnegative super-first Chern class when 
\begin{equation} 
\label{c_1(WCP)>=0} 
\sum_{i=1}^m Q_i - \sum_{a=1}^n q_a \ge 0 \, . 
\end{equation} 
We will assume that a crepant resolution exists, the condition 
(\ref{c_1(WCP)>=0}) is satisfied, and the bosonic and fermionic weights are  
positive integers. 
The K\"{a}hler supermanifold that we obtain can be described as the resolved target  
space of a nonlinear sigma model phase of a $ (2,2) $ supersymmetric $ U(1) $ gauged 
linear sigma model with $ m $ bosonic chiral superfields $ \Phi_i $  
having $ U(1) $ charges $ Q_i \, , $ $ i=1,\ldots,m \, , $ and $ n $ fermionic  
chiral superfields $ \Xi_a $ having $ U(1) $ charges $ q_a \, , $  
$ a = 1,\ldots,n \, . $  
The classical Lagrangian is \cite{SekiSugiyama:Gauged} 
\begin{align} 
\label{L_WCP} 
L_{\mathbf{WCP}^{m-1|n}_{(Q_1,\ldots,Q_m|q_1,\ldots,q_n)}} 
 &= \int d^4 \theta 
    \left( 
            \sum_{i=1}^m \overline{\Phi}_i e^{2 Q_i V} \Phi_i 
          + \sum_{a=1}^n \overline{\Xi}_a e^{2 q_a V} \Xi_a 
          - \frac{1}{2 e^2} \overline{\Sigma} \Sigma 
    \right) 
\nonumber 
\\ 
  &\phantom{ =\ } 
  - \frac{1}{2} 
    \left( 
           t \int d^2 \tilde{\theta} \, \Sigma + c.c. 
    \right), 
\end{align} 
where $ \Sigma = \overline{D}_{+} D_{-} V $ is the twisted chiral field  
strength of the $ U(1) $ vector superfield $ V $, $ e $ is the $ U(1) $ gauge  
coupling, and $ t = r - i \vartheta $ is the complexified Fayet-Iliopoulos 
parameter. 
The nonlinear sigma model phase is realized in the low energy limit with  
$ r >> 0 $ with $ \sigma = 0 $ and target space   
\begin{equation} 
     \frac{ 
            \left\{ ( \phi_1,\ldots,\phi_m | \xi_1,\ldots,\xi_n ) 
                    \left| 
                           \sum_{i=1}^m Q_i \left| \phi_i \right |^2 
                         + \sum_{a=1}^n q_a \left| \xi_a \right |^2 
                         = r 
                    \right. 
            \right\} 
          }{U(1)} \, . 
\end{equation} 
Here, $ \sigma $, $ \phi_i $ and $ \xi_a $ are respectively the lowest 
components of $ \Sigma $, $ \Phi_i $,  and $ \Xi_a $. 
 
The gauge theory is super-renormalizable with respect to $ e $.  
When the super Calabi-Yau condition 
\begin{equation} 
\label{SCYcondition} 
\sum_{i=1}^m Q_i - \sum_{a=1}^n q_a = 0 
\end{equation} 
is not satisfied, to cancel a one-loop ultraviolet divergence, the  
Fayet-Iliopoulos parameter $ r $ must be renormalized as 
\begin{equation*} 
r(\mu) = r(\Lambda_{UV}) 
       + \left( 
                \sum_{i=1}^m Q_i - \sum_{a=1}^n q_a 
         \right) 
         \ln{ 
              \left( 
                     \frac{\mu}{\Lambda_{UV}} 
              \right) 
            } \, , 
\end{equation*} 
where $ r(\mu) $ is the renormalized Fayet-Iliopoulos parameter at the scale  
$ \mu $, $ \Lambda_{UV} $ is the ultraviolet cutoff, and 
\begin{equation*} 
r(\Lambda_{UV}) 
   = \left( 
            \sum_{i=1}^m Q_i - \sum_{a=1}^n q_a 
     \right) 
     \ln{ 
          \left( 
                 \frac{\Lambda_{UV}}{\Lambda} 
          \right)  
        } \, . 
\end{equation*} 
In this case, the dimensionless parameter $ r $ of the classical theory  
is replaced by the renormalization group invariant dynamical scale $ \Lambda $ 
in the quantum theory. 
 
The classical theory possesses a $ U(1)_V $ vector R-symmetry when $ \Sigma $ 
is assigned vector charge $ 0 $.  It also has a $ U(1)_A $ axial R-symmetry 
when $ \Sigma $ is assigned axial charge $ 2 $. 
The vector R-symmetry is an exact symmetry of the quantum theory, but the  
axial R-symmetry is subject to a chiral anomaly. 
An axial rotation by angle $ \alpha $ shifts the theta angle as  
\begin{equation*} 
\vartheta  
   \rightarrow  \vartheta  
   - 2 \left( 
              \sum_{i=1}^m Q_i - \sum_{a=1}^n q_a 
       \right) \alpha \, . 
\end{equation*} 
Note that the axial anomaly vanishes when the super Calabi-Yau condition 
(\ref{SCYcondition}) is satisfied. 
 
Following the arguments of \cite{AganagicVafa:Mirror}, we obtain the super  
Landau-Ginzburg mirror period 
\begin{align} 
\label{Pi_WCP} 
\Pi &_{ \widetilde{\mathbf{WCP}}^{m-1|n}_{(Q_1,\ldots,Q_m|q_1,\ldots,q_n)} } 
\nonumber 
\\ 
 &=\int  
   \left( 
         \prod_{i=1}^m \, dY_i  
   \right) 
   \left(    
      \prod_{a=1}^n \, dX_a \, d \eta_a \, d \gamma_a  
   \right) 
   \delta \left( 
                \sum_{i=1}^m Q_i Y_i - \sum_{a=1}^n q_a X_a - t 
          \right) 
\nonumber 
\\ 
&\phantom{=\int} \times 
  \exp{ \left[ 
              - \sum_{i=1}^m  e^{-Y_i}  
              - \sum_{a=1}^n e^{-X_a}(1 + \eta_a \gamma_a) 
        \right]  
      }, 
\end{align} 
where $ Y_i $ and $ X_a $ are twisted chiral superfields (with periodicity 
$ 2 \pi i $) which satisfy 
\begin{equation} 
\mathrm{Re} \, Y_i =  \overline{\Phi}_i e^{2 Q_i V} \Phi_i \, , 
\qquad 
\mathrm{Re} \, X_a = - \overline{\Xi}_a e^{2 q_a V} \Xi_a 
\end{equation} 
and $ \eta_a $ and $ \gamma_a $ are fermionic superfields. 
In fact, (\ref{Pi_WCP}) holds whenever (\ref{c_1(WCP)>=0}) is satisfied.

\section{\label{Tot}\texorpdfstring{$ \mathrm{Tot} \left( \mathcal{O}(-s)  
\rightarrow \mathbf{WCP}^{m-1|n} \right) $}{TotO(-s)-->WCPm-1n}} 
 
Consider the noncompact supervariety  
$  
\mathrm{Tot}  
\left(  
       \mathcal{O}(-s) \rightarrow  
       \mathbf{WCP}^{m-1|n}_{ ( Q_1,\ldots,Q_m | q_1,\ldots,q_n ) }  
\right)  
$, i.e. the total space of the line bundle $ \mathcal{O}(-s) $ over the  
weighted complex superprojective space of Section \ref{WCP}. 
A K\"{a}hler supermanifold may be obtained as a crepant resolution of this  
supervariety if a crepant resolution exists. 
Such a supermanifold would have nonnegative super-first Chern class when 
\begin{equation} 
\label{c_1(Tot)>=0} 
\sum_{i=1}^m Q_i - \sum_{a=1}^n q_a -s \ge 0 \, . 
\end{equation} 
We will assume that a crepant resolution exists, the condition  
(\ref{c_1(Tot)>=0}) is satisfied, and $ s $ is a positive integer. 
To describe our K\"{a}hler supermanifold in terms of a gauged linear sigma  
model, we add to (\ref{L_WCP}) a kinetic term for a bosonic chiral superfield 
$ P $ of $ U(1) $ charge $ -s $.   
The resulting Lagrangian is 
\begin{align} 
\label{L_Tot} 
L &_{\mathrm{Tot}  
     \left(  
           \mathcal{O}(-s) \rightarrow 
           \mathbf{WCP}^{m-1|n}_{(Q_1,\ldots,Q_m|q_1,\ldots,q_n)} 
     \right) 
  } 
\nonumber 
\\     
  &\phantom{\mathrm{Tot}} 
   = \int d^4 \theta 
     \left( 
             \sum_{i=1}^m \overline{\Phi}_i e^{2 Q_i V} \Phi_i 
           + \overline{P} e^{-2 s V} P 
           + \sum_{a=1}^n \overline{\Xi}_a e^{2 q_a V} \Xi_a 
           - \frac{1}{2 e^2} \overline{\Sigma} \Sigma 
     \right) 
\nonumber 
\\ 
  &\phantom{ \mathrm{Tot} = } 
   - \frac{1}{2} 
     \left( 
            t \int d^2 \tilde{\theta} \,\Sigma + c.c. 
     \right). 
\end{align} 
The nonlinear sigma model phase is realized in the low energy limit with  
$ r \gg 0 \, , $ $ \sigma = 0 \, , $ and target space 
\begin{equation} 
\frac{ 
\left\{ ( \phi_1,\ldots,\phi_m,p | \xi_1,\ldots,\xi_n ) 
        \left| 
                \textstyle{ 
                \sum_{i=1}^m Q_i \left| \phi_i \right |^2 
              - s | p |^2  
              + \sum_{a=1}^n q_a \left| \xi_a \right |^2 
                          } 
              = r 
        \right. 
\right\} 
} 
{ U(1) } \, , 
\end{equation} 
where $ p $ is the lowest component of $ P $. 
If the super Calabi-Yau condition 
\begin{equation} 
\sum_{i=1}^m Q_i - \sum_{a=1}^n q_a = s 
\end{equation} 
is satisfied, then the Fayet-Iliopoulos parameter $ r $ does not renormalize. 
 
Extending the result (\ref{Pi_WCP}) to the present case, we obtain for    
the super Landau-Ginzburg mirror period  
\begin{align} 
\label{Pi_Tot} 
\Pi &_{ \widetilde{\mathrm{Tot}} 
        \left( 
               \mathcal{O}(-s) \rightarrow 
               \mathbf{WCP}^{m-1|n}_{(Q_1,\ldots,Q_m|q_1,\ldots,q_n)} 
        \right)  
      } 
\nonumber 
\\  
    &\phantom{\widetilde{\mathrm{Tot}}} 
     =\int 
      \left( 
            \prod_{i=1}^m dY_i 
      \right) 
      dY_P  
      \left( 
            \prod_{a=1}^n dX_a \, d\eta_a \, d \gamma_a 
      \right) 
\delta \left( 
                \sum_{i=1}^m Q_i Y_i - s Y_P 
              - \sum_{a=1}^n q_a X_a - t 
          \right) 
\nonumber 
\\ 
&\phantom{\widetilde{\mathrm{Tot}}=\int} \times 
  \exp{ \left[ 
              - \sum_{i=1}^m  e^{-Y_i} - e^{-Y_P} 
              - \sum_{a=1}^n e^{-X_a}(1 + \eta_a \gamma_a) 
        \right] 
      }, 
\end{align} 
where $ Y_P $ is a bosonic twisted chiral superfield (with periodicity  
$ 2 \pi i $) which satisfies 
\begin{equation} 
\mathrm{Re} \, Y_P = \overline{P} e^{-2 s V} P \, . 
\end{equation} 
This result will be useful in the next section where we study an  
associated compact theory.

\section{\label{WCP[s]}Hypersurface in  
\texorpdfstring{$ \mathbf{WCP}^{m-1|n} $}{WCPm-1|n}} 
 
Consider a member of the family of compact supervarieties  
$ \mathbf{WCP}^{m-1|n}_{ ( Q_1,\ldots,Q_m | q_1,\ldots,q_n ) [s] } $, i.e. 
a hypersurface defined by the zero locus of a quasihomogeneous  
polynomial $ G = G(\phi, \xi) $ of degree $ s $ in the weighted complex 
superprojective space of Section \ref{WCP}. 
A K\"{a}hler supermanifold may be obtained as a crepant resolution of this  
supervariety if a crepant resolution exists. 
Such a supermanifold would have nonnegative super-first Chern class when 
\begin{equation} 
\label{c_1(WCP[s])>=0} 
\sum_{i=1}^m Q_i - \sum_{a=1}^n q_a -s \ge 0 \, . 
\end{equation} 
We will assume that a crepant resolution exists and the condition  
(\ref{c_1(WCP[s])>=0}) is satisfied. 
To describe our K\"{a}hler supermanifold in terms of a gauged linear sigma 
model, we add to (\ref{L_Tot}) an untwisted F-term with  
superpotential of the form $ P \cdot G(\Phi,\Xi) $, where $ G(\Phi,\Xi) $ is 
a quasihomogeneous polynomial of $ U(1) $ charge $ s $.   
The resulting Lagrangian is 
\begin{align} 
\label{L_WCP[s]} 
L &_{ \mathbf{WCP}^{m-1|n}_{ ( Q_1,\ldots,Q_m | q_1,\ldots,q_n ) [s] } } 
\nonumber
\\ 
  &\phantom{\mathbf{WCP}}
   = \int d^4 \theta 
     \left( 
             \sum_{i=1}^m \overline{\Phi}_i e^{2 Q_i V} \Phi_i 
           + \sum_{a=1}^n \overline{\Xi}_a e^{2 q_a V} \Xi_a 
           + \overline{P} e^{-2 s V} P
           - \frac{1}{2 e^2} \overline{\Sigma} \Sigma 
     \right) 
\nonumber 
\\ 
  &\phantom{= \mathbf{WCP}}
   - \frac{1}{2} 
     \left( 
            t \int d^2 \tilde{\theta} \,\Sigma + c.c. 
     \right) 
   + \left( 
            \int d^2 \theta \,  
            P \cdot  
            G(\Phi,\Xi) 
          + c.c.  
     \right). 
\end{align} 
The nonlinear sigma model phase is realized in the low energy limit  
with $ r \gg 0 \, , $ $ \sigma = 0 \, , $ $ p = 0 \, , $ and target space 
\begin{equation} 
\left\{  
        G = 0 
\right\} 
   \subset 
     \frac{ 
            \left\{ ( \phi_1,\ldots,\phi_m | \xi_1,\ldots,\xi_n ) 
               \left| 
                  \textstyle{ 
                              \sum_{i=1}^m Q_i \left| \phi_i \right |^2 
                             + \sum_{a=1}^n q_a \left| \xi_a \right |^2 
                            } 
                             = r 
                    \right. 
            \right\} 
          }{U(1)} \, . 
\end{equation} 
If the super Calabi-Yau condition 
\begin{equation} 
\label{CompactSCYH-Condition} 
\sum_{i=1}^m Q_i - \sum_{a=1}^n q_a = s 
\end{equation} 
is satisfied, then the Fayet-Iliopoulos parameter $ r $ does not renormalize. 
 
Generalizing the discussion in \cite{HoriVafa:Mirror}, we can obtain the super 
Landau-Ginzburg mirror period by allowing the operator  
$ - s \, \partial / \partial t $ to act on the period given by  
(\ref{Pi_Tot}), i.e. 
\begin{align} 
\label{Pi_WCP[s]} 
\Pi &_{  
        \widetilde{\mathbf{WCP}}^{m-1|n} 
        _{ ( Q_1,\ldots,Q_m | q_1,\ldots,q_n ) [s] }   
      } 
\nonumber 
\\ 
   &=-s \frac{ \partial}{ \partial t}  
      \Pi_{ \widetilde{\mathrm{Tot}} 
            \left( 
                   \mathcal{O}(-s) \rightarrow 
                   \mathbf{WCP}^{m-1|n}_{(Q_1,\ldots,Q_m|q_1,\ldots,q_n)} 
            \right) 
          } 
\nonumber 
\\ 
   &= \int  
      \left( 
            \prod_{i=1}^m dY_i 
      \right) 
      dY_P \, e^{-Y_P}   
      \left( 
            \prod_{a=1}^n dX_a \, d\eta_a \, d \gamma_a 
      \right) 
\delta \left( 
                \sum_{i=1}^m Q_i Y_i - s Y_P    
              - \sum_{a=1}^n q_a X_a - t 
       \right) 
\nonumber 
\\ 
 &\phantom{= \int} 
  \times 
  \exp{ \left[ 
              - \sum_{i=1}^m  e^{-Y_i} -  e^{-Y_P} 
              - \sum_{a=1}^n e^{-X_a}(1 + \eta_a \gamma_a) 
        \right] 
      }.  
\end{align} 
 
Integrating (\ref{Pi_WCP[s]}) over $ Y_P $ yields 
\begin{align*} 
\Pi &_{  
        \widetilde{\mathbf{WCP}}^{m-1|n} 
        _{ ( Q_1,\ldots,Q_m | q_1,\ldots,q_n ) [s]} 
      } 
\nonumber 
\\  
  &= 
     \int 
     \left( 
            \prod_{i=1}^m dY_i 
     \right) 
     \left[ 
            e^{ \frac{t}{s} } 
            \prod_{i=1}^m \left( e^{-Y_i} \right)^{ \frac{Q_i}{s} } 
            \prod_{a=1}^n \left( e^{-X_a} \right)^{- \frac{q_a}{s}} 
     \right]  
     \left( 
            \prod_{a=1}^n dX_a \, d\eta_a \, d \gamma_a 
     \right)   
\\ 
    &\phantom{= \int} 
     \times 
     \exp{ \left[ 
                - \sum_{i=1}^m  e^{-Y_i}  
                - e^{ \frac{t}{s} } 
                  \prod_{i=1}^m \left( e^{-Y_i} \right)^{ \frac{Q_i}{s} } 
                  \prod_{a=1}^n \left( e^{-X_a} \right)^{- \frac{q_a}{s}} 
               - \sum_{a=1}^n e^{-X_a} 
                  \left( 1 + \eta_a \gamma_a \right) 
           \right] 
          }. 
\end{align*} 
Suppose there exist invertible matrices $ (M_{ji}) $ and $ (N_{ba}) $ of  
nonnegative integers such that 
\begin{equation} 
\label{MandNequations} 
s = \sum_{i=1}^m M_{ji} Q_i = \sum_{a=1}^n N_{ba} q_a \, , 
\qquad  
j = 1,\ldots,m \, ;  
\quad 
\ b = 1,\ldots,n \, . 
\end{equation} 
Now, consider the change of variables  
\begin{gather} 
e^{-Y_i} = \prod_{j=1}^m y_j^{ M_{j i} } \, , 
\qquad 
e^{-X_a} = \prod_{b=1}^{n} x_b^{N_{ba}} \, , 
\qquad 
\eta_a = x_a^{-1} \hat{\eta}_a \, , 
\qquad 
\gamma_a = x_a^{-1} \hat{\gamma}_a \, . 
\end{gather} 
This change of variables one-to-one up to the action of the group $ \Gamma $  
defined by 
\begin{equation} 
\label{Gamma} 
\Gamma : \ 
y_j \rightarrow \omega_{y_j} y_j \, , 
\qquad 
x_b \rightarrow \omega_{x_b} x_b \, , 
\qquad 
\hat{\eta}_a \rightarrow \omega_{x_a} \hat{\eta}_a \, , 
\qquad 
\hat{\gamma}_a \rightarrow \omega_{x_a} \hat{\gamma}_a \, , 
\end{equation} 
such that 
\begin{equation*} 
\prod_{j=1}^m \omega_{y_j}^{M_{ji}} = 1 \, , 
\qquad 
\prod_{b=1}^n \omega_{x_b}^{N_{ba}} = 1 \, , 
\qquad 
\prod_{j=1}^m \omega_{y_j} \prod_{b=1}^n \omega_{x_b}^{-1} = 1 \, . 
\end{equation*} 
In terms of the new variables, we obtain 
\begin{align} 
\Pi &_{  
        \widetilde{\mathbf{WCP}}^{m-1|n} 
        _{ ( Q_1,\ldots,Q_m | q_1,\ldots,q_n ) [s] } 
      } 
\nonumber 
\\ 
 &=    (-1)^{m + n}  
       \det{(M_{ji})} \det{(N_{ba})} \, e^{t/s} 
       \int 
       \left( 
              \prod_{i=1}^m dy_i 
       \right) 
       \left( 
              \prod_{a=1}^n dx_a \, d\hat{\eta}_a \, d\hat{\gamma}_a 
       \right) 
\nonumber 
\\ 
 &\phantom{=(} \times 
  \exp{ \left[ 
            - \sum_{i=1}^m  \prod_{j=1}^m y_j^{ M_{ji} } 
            - e^{t/s} \prod_{j=1}^m y_j \prod_{b=1}^n x_b^{-1} 
            - \sum_{a=1}^n  
              \left( 
                     1 + x_a^{-2} \hat{\eta}_a \hat{\gamma}_a 
              \right) 
              \prod_{b=1}^n x_b^{ N_{ba} } 
        \right] 
      }. 
\end{align} 
This is the period for the super Landau-Ginzburg orbifold  
$$ \widetilde{W} / \Gamma \, $$ where  
\begin{equation} 
\label{Wtilde} 
\widetilde{W}  
   = \sum_{i=1}^m  \prod_{j=1}^m y_j^{ M_{ji} } 
    + e^{t/s} \prod_{j=1}^m y_j \prod_{b=1}^n x_b^{-1} 
    + \sum_{a=1}^n  
       \left( 
             1 + x_a^{-2} \hat{\eta}_a \hat{\gamma}_a 
      \right) 
      \prod_{b=1}^n x_b^{ N_{ba} } 
\end{equation} 
and $ \Gamma $ is given by (\ref{Gamma}). 
 
As proven in the Appendix, $ \widetilde{W} $ is quasihomogeneous of some degree 
$ s^{\prime} $ for all values of $ t $ if and only if the super Calabi-Yau condition  
(\ref{CompactSCYH-Condition}) is satisfied. 
In this case, for appropriately chosen fermionic weights, the super  
Landau-Ginzburg orbifold $ \widetilde{W}/ \Gamma $ corresponds to a Calabi-Yau 
supermanifold obtained as a crepant resolution of 
\begin{equation} 
\label{WCP[s]mirror} 
\frac{ \left \{ \widetilde{W} = 0 \right \} }{\Gamma / \widetilde{J}} 
   \subset 
   \mathbf{WCP}^{m+n-1|2n}_{ 
   ( n_{y_1},\ldots,n_{y_m}, n_{x_1},\ldots,n_{x_n} | 
     n_{ \hat{\eta}_1 }, n_{ \hat{\gamma}_1 }, \ldots,  
     n_{ \hat{\eta}_n }, n_{ \hat{\gamma}_n } ) 
                           } \, , 
\end{equation} 
where $ \widetilde{J} $ is the diagonal subgroup of the phase symmetries of  
$ \widetilde{W} $, i.e.  
\begin{equation} 
\widetilde{J}:  
\quad 
y_j \rightarrow e^{ 2 \pi i \, n_{y_j} / s^{\prime} } y_j \, , 
\quad 
x_b \rightarrow e^{ 2 \pi i \, n_{x_b} / s^{\prime} } x_b \, , 
\quad 
\hat{\eta}_b \rightarrow e^{ 2 \pi i \, n_{ \hat{\eta}_b } / s^{\prime} } \hat{\eta}_b  
\, , 
\quad 
\hat{\gamma}_b \rightarrow e^{ 2 \pi i \, n_{ \hat{\gamma}_b } / s^{\prime} }  
\hat{\gamma}_b \, .  
\end{equation} 
\begin{example} 
\label{Counterexample} 
 
Consider the Landau-Ginzburg orbifold $ G_{bos} / J_{bos} $, where 
\begin{equation*} 
G_{bos} =   \phi_1^6 + \phi_2^6 + \phi_3^6 + \phi_4^3 + \phi_5^3 
          + \phi_6^3 + \phi_7^2 
\end{equation*} 
and $ J_{bos} $ is the diagonal subgroup of the phase symmetries of 
$ G_{bos} $, i.e. 
\begin{equation*} 
J_{bos}: \ \phi_i \rightarrow e^{2 \pi i Q_i / 6}  \phi_i \, . 
\end{equation*} 
Using the techniques of \cite{LGOHodge}, we find that the Hodge diamond of 
$ G_{bos} / J_{bos} $ is  
\begin{equation*} 
\begin{matrix} 
& & & 1 & & & \\ 
& & 0 & & 0 & & \\ 
& 0 & & 0 & & 0 & \\ 
1 & & 84 & & 84 & & 1 \\ 
& 0 & & 0 & & 0 & \\ 
& & 0 & & 0 & & \\ 
& & & 1 & & & \\ 
\end{matrix} \, . 
\end{equation*} 
The proposal of \cite{Sethi} predicts that, for appropriately chosen fermionic 
weights $ (q_1,q_2) = (q_{1*},q_{2*}) $, $ G_{bos} / J_{bos} $ corresponds to  
a Calabi-Yau supermanifold $ M^{(q_1,q_2)} $ obtained as a crepant resolution of 
\begin{equation*} 
\{ G^{(q_1,q_2)} = 0 \} \in  
\mathbf{WCP}^{6|2}_{ (1,1,1,2,2,2,3 | q_1, q_2 ) [6] }  
\end{equation*} 
where 
\begin{equation*} 
G^{(q_1,q_2)} = G_{bos} + \xi_1 \xi_2 \, . 
\end{equation*} 
The positive integer values of $ (q_1,q_2) $ consistent with the  
quasihomogenity constraint 
\begin{equation*} 
q_1 + q_2 = 6 
\end{equation*} 
are (up to a relabelling of $ \xi_1 $ and $ \xi_2 $)  
\begin{equation*} 
(q_1,q_2) \in A = \{(1,5),(2,4),(3,3) \}.  
\end{equation*} 
The solutions for $ (q_{1*},q_{2*}) $ are those values of $ (q_1,q_2) \in A $ 
for which the Hodge diamond of $ M^{(q_1,q_2)} $ agrees with the Hodge diamond 
of $ G_{bos} / J_{bos} $. 
It was found in \cite{GaravusoKreuzerNoll:Fano} that the Hodge numbers obtained 
for $ M^{(q_1,q_2)} $ by using the heuristic arguments of \cite{Sethi} do not  
agree with the Hodge numbers of $ G_{bos} / J_{bos} $ for any  
$ (q_1,q_2) \in A $.    
We can use the recently developed techniques of  
\cite{GaravusoKatzarkovNoll} to properly compute the Hodge diamond of  
$ M^{(q_1,q_2)} $ for each $ (q_1,q_2) \in A $. 
 
According to (\ref{WCP[s]mirror}), given $ M^{(q_{1*},q_{2*})} $, for  
appropriately chosen fermionic weights 
\begin{equation*} 
(n_{\hat{\eta}_1}, n_{\hat{\gamma}_1}, n_{\hat{\eta}_2}, n_{\hat{\gamma}_2}) 
= 
(n_{\hat{\eta}_1*}, n_{\hat{\gamma}_1*}, n_{\hat{\eta}_2*}, n_{\hat{\gamma}_2*}) 
\, , 
\end{equation*} 
a Calabi-Yau supermanifold 
$  
\widetilde{M}^{ 
               (n_{\hat{\eta}_1}, n_{\hat{\gamma}_1}, 
                n_{\hat{\eta}_2}, n_{\hat{\gamma}_2}) 
              }_{(q_{1*},q_{2*})} 
$ 
obtained as a crepant resolution of 
\begin{equation*} 
     \frac{ 
            \left\{ 
                    \widetilde{W}^{(q_{1*},q_{2*})} = 0 
            \right\} 
          } 
          { \Gamma^{(q_{1*},q_{2*})} / \widetilde{J}^{(q_{1*},q_{2*})} } 
\subset 
\mathbf{WCP}^{8|4}_{ 
                     \left( 
                           \left.    
                                  1,1,1,2,2,2,3, 
                                  n_{x_1}^{(q_{1*},q_{2*})}, 
                                  n_{x_2}^{(q_{1*},q_{2*})}  
                           \right| 
                           n_{ \hat{\eta}_1 }, n_{ \hat{\gamma}_1 }, 
                           n_{ \hat{\eta}_2 }, n_{ \hat{\gamma}_2 }   
                     \right) 
                   } 
\end{equation*} 
is mirror to $ M^{(q_{1*},q_{2*})} $. 
From (\ref{Wtilde}), we have that 
\begin{equation*} 
\widetilde{W}^{(q_{1*},q_{2*})} 
   = \sum_{i=1}^7  \prod_{j=1}^7 y_j^{ M_{ji} } 
    + e^{t/6} \prod_{j=1}^7 y_j \prod_{b=1}^2 x_b^{-1} 
    + \sum_{a=1}^2 
       \left( 
             1 + x_a^{-2} \hat{\eta}_a \hat{\gamma}_a 
      \right) 
      \prod_{b=1}^2 x_b^{ N_{ba}^{(q_{1*},q_{2*})} } \, , 
\end{equation*} 
where the matrix elements of the invertible matrices  
$ (M_{ji}) $ and $ (N_{ba})^{(q_{1*},q_{2*})} $ of nonnegative integers are  
respectively given by 
\begin{equation*} 
(M_{ji}) = \mathrm{diag}(6,6,6,3,3,3,2)_{ji} \, , 
\qquad 
6 = \sum_{a=1}^2 N_{ba}^{(q_{1*},q_{2*})} q_{a*}  
\end{equation*} 
as required by (\ref{MandNequations}). 
Let $ B^{(q_1,q_2)} $ be the set of values of  
$  
(n_{\hat{\eta}_1}, n_{\hat{\gamma}_1},  
 n_{\hat{\eta}_2}, n_{\hat{\gamma}_2})  
$ 
for which $ \widetilde{W}^{(q_1,q_2)} $ is quasihomogeneous of degree 6. 
The solutions for  
$ 
(n_{\hat{\eta}_1*}, n_{\hat{\gamma}_1*}, 
 n_{\hat{\eta}_2*}, n_{\hat{\gamma}_2*})   
$ 
are those values of 
$ 
(n_{\hat{\eta}_1}, n_{\hat{\gamma}_1}, 
 n_{\hat{\eta}_2}, n_{\hat{\gamma}_2}) \in B^{(q_{1*},q_{2*})}  
$ 
for which the Hodge diamond of  
$ 
\widetilde{M}^{ 
               (n_{\hat{\eta}_1}, n_{\hat{\gamma}_1}, 
                n_{\hat{\eta}_2}, n_{\hat{\gamma}_2}) 
              }_{(q_{1*},q_{2*})} 
$ 
is what is expected for a mirror of $ M^{(q_{1*},q_{2*})} $, i.e. 
\begin{equation*} 
\begin{matrix} 
& & & 1 & & & \\ 
& & 0 & & 0 & & \\ 
& 0 & & 84 & & 0 & \\ 
1 & & 0 & & 0 & & 1 \\ 
& 0 & & 84 & & 0 & \\ 
& & 0 & & 0 & & \\ 
& & & 1 & & & \\ 
\end{matrix} \, . 
\end{equation*} 
This is also what is expected for the Hodge diamond of the super  
Landau-Ginzburg orbifold 
$ \widetilde{W}^{(q_{1*},q_{2*})} / \Gamma^{(q_{1*},q_{2*})} $. 
We can use the recently developed techniques of  
\cite{GaravusoKatzarkovNoll} to compute the Hodge diamond of  
$ 
\widetilde{M}^{ 
               (n_{\hat{\eta}_1}, n_{\hat{\gamma}_1}, 
                n_{\hat{\eta}_2}, n_{\hat{\gamma}_2}) 
              }_{(q_{1*},q_{2*})} 
$ 
for each  
$ 
(n_{\hat{\eta}_1}, n_{\hat{\gamma}_1},  
 n_{\hat{\eta}_2}, n_{\hat{\gamma}_2}) \in B^{(q_{1*},q_{2*})}   
$.

\end{example} 
 
\begin{example} 
\label{ExHyper2} 
 
Consider the Landau-Ginzburg orbifold $ G_{bos} / J_{bos} $, where 
\begin{equation*} 
G_{bos} =   \phi_1^6 + \phi_2^6 + \phi_3^3 + \phi_4^3 + \phi_5^3 
          + \phi_6^3 + \phi_7^3 
\end{equation*} 
and $ J_{bos} $ is the diagonal subgroup of the phase symmetries of 
$ G_{bos} $, i.e. 
\begin{equation*} 
J_{bos}: \ \phi_i \rightarrow e^{2 \pi i Q_i / 6}  \phi_i \, . 
\end{equation*} 
Using the techniques of \cite{LGOHodge}, we find that the Hodge diamond of 
$ G_{bos} / J_{bos} $ is 
\begin{equation*} 
\begin{matrix} 
& & & 1 & & & \\ 
& & 0 & & 0 & & \\ 
& 0 & & 1 & & 0 & \\ 
1 & & 73 & & 73 & & 1 \\ 
& 0 & & 1 & & 0 & \\ 
& & 0 & & 0 & & \\ 
& & & 1 & & & \\ 
\end{matrix} \, . 
\end{equation*} 
The proposal of \cite{Sethi} predicts that, for appropriately chosen fermionic 
weights $ (q_1,q_2) = (q_{1*},q_{2*}) $, $ G_{bos} / J_{bos} $ corresponds to  
a Calabi-Yau supermanifold $ M^{(q_1,q_2)} $ obtained as a crepant resolution of 
\begin{equation*} 
\{ G^{(q_1,q_2)} = 0 \} \in 
\mathbf{WCP}^{6|2}_{ (1,1,2,2,2,2,2 | q_1, q_2 ) [6] } 
\end{equation*} 
where 
\begin{equation*} 
G^{(q_1,q_2)} = G_{bos} + \xi_1 \xi_2 \, . 
\end{equation*} 
The positive integer values of $ (q_1,q_2) $ consistent with the  
quasihomogeneity constraint 
\begin{equation*} 
q_1 + q_2 = 6 
\end{equation*} 
are (up to a relabelling of $ \xi_1 $ and $ \xi_2 $) 
\begin{equation*} 
(q_1,q_2) \in A = \{(1,5),(2,4),(3,3) \}. 
\end{equation*} 
The solutions for $ (q_{1*},q_{2*}) $ are those values of $ (q_1,q_2) \in A $ 
for which the Hodge diamond of $ M^{(q_1,q_2)} $ agrees with the Hodge diamond 
of $ G_{bos} / J_{bos} $. 
It was found in \cite{Sethi,GaravusoKreuzerNoll:Fano} that the Hodge numbers  
obtained for $ M^{(q_1,q_2)} $ by using the heuristic arguments of  
\cite{Sethi} agree with the Hodge numbers of $ G_{bos} / J_{bos} $ when 
$ (q_1,q_2) = (2,4) $ but disagree when $ (q_1,q_2) \in \{(2,4),(3,3) \} $.  
We can use the recently developed techniques of 
\cite{GaravusoKatzarkovNoll} to properly compute the Hodge diamond of 
$ M^{(q_1,q_2)} $ for each $ (q_1,q_2) \in A $. 
 
According to (\ref{WCP[s]mirror}), given $ M^{(q_{1*},q_{2*})} $, for 
appropriately chosen fermionic weights 
\begin{equation*} 
(n_{\hat{\eta}_1}, n_{\hat{\gamma}_1}, n_{\hat{\eta}_2}, n_{\hat{\gamma}_2}) 
= 
(n_{\hat{\eta}_1*}, n_{\hat{\gamma}_1*}, n_{\hat{\eta}_2*}, n_{\hat{\gamma}_2*}) 
\, , 
\end{equation*} 
a Calabi-Yau supermanifold 
$  
\widetilde{M}^{ 
               (n_{\hat{\eta}_1}, n_{\hat{\gamma}_1}, 
                n_{\hat{\eta}_2}, n_{\hat{\gamma}_2}) 
              }_{(q_{1*},q_{2*})} 
$ 
obtained as a crepant resolution of 
\begin{equation*} 
\frac{ 
       \left\{ 
               \widetilde{W}^{(q_{1*},q_{2*})} = 0 
       \right\} 
     } 
     { \Gamma / \widetilde{J}^{(q_{1*},q_{2*})} } 
\subset 
\mathbf{WCP}^{8|4}_{ 
                     \left(  
                            \left. 
                                   1,1,2,2,2,2,2, 
                                   n_{x_1}^{(q_{1*},q_{2*})}, 
                                   n_{x_2}^{(q_{1*},q_{2*})}  
                            \right| 
                            n_{ \hat{\eta}_1 }, n_{ \hat{\gamma}_1 },  
                            n_{ \hat{\eta}_2 }, n_{ \hat{\gamma}_2 }  
                     \right) 
                   } 
\end{equation*} 
is mirror to $ M^{(q_{1*},q_{2*})} $. 
From (\ref{Wtilde}), we have that 
\begin{equation*} 
\widetilde{W}^{(q_{1*},q_{2*})} 
   = \sum_{i=1}^7  \prod_{j=1}^7 y_j^{ M_{ji} } 
    + e^{t/6} \prod_{j=1}^7 y_j \prod_{b=1}^2 x_b^{-1} 
    + \sum_{a=1}^2 
       \left( 
             1 + x_a^{-2} \hat{\eta}_a \hat{\gamma}_a 
      \right) 
      \prod_{b=1}^2 x_b^{ N_{ba}^{(q_{1*},q_{2*})} } \, , 
\end{equation*} 
where the matrix elements of the invertible matrices 
$ (M_{ji}) $ and $ (N_{ba})^{(q_{1*},q_{2*})} $ of nonnegative integers are 
respectively given by 
\begin{equation*} 
(M_{ji}) = \mathrm{diag}(6,6,3,3,3,3,3)_{ji} \, , 
\qquad 
6 = \sum_{a=1}^2 N_{ba}^{(q_{1*},q_{2*})} q_{a*} 
\end{equation*} 
as required by (\ref{MandNequations}). 
Let $ B^{(q_1,q_2)} $ be the set of values of 
$ 
(n_{\hat{\eta}_1}, n_{\hat{\gamma}_1}, 
 n_{\hat{\eta}_2}, n_{\hat{\gamma}_2}) 
$ 
for which $ \widetilde{W}^{(q_1,q_2)} $ is quasihomogeneous of degree 6. 
The solutions for 
$ 
(n_{\hat{\eta}_1*}, n_{\hat{\gamma}_1*}, 
 n_{\hat{\eta}_2*}, n_{\hat{\gamma}_2*}) 
$ 
are those values of 
$ 
(n_{\hat{\eta}_1}, n_{\hat{\gamma}_1}, 
 n_{\hat{\eta}_2}, n_{\hat{\gamma}_2}) \in B^{(q_{1*},q_{2*})} 
$ 
for which the Hodge diamond of 
$ 
\widetilde{M}^{ 
               (n_{\hat{\eta}_1}, n_{\hat{\gamma}_1}, 
                n_{\hat{\eta}_2}, n_{\hat{\gamma}_2}) 
              }_{(q_{1*},q_{2*})} 
$ 
is what is expected for a mirror of $ M^{(q_{1*},q_{2*})} $, i.e. 
\begin{equation*} 
\begin{matrix} 
& & & 1 & & & \\ 
& & 0 & & 0 & & \\ 
& 0 & & 73 & & 0 & \\ 
1 & & 1 & & 1 & & 1 \\ 
& 0 & & 73 & & 0 & \\ 
& & 0 & & 0 & & \\ 
& & & 1 & & & \\ 
\end{matrix} \, . 
\end{equation*} 
This is also what is expected for the Hodge diamond of the super  
Landau-Ginzburg orbifold 
$ \widetilde{W}^{(q_{1*},q_{2*})} / \Gamma^{(q_{1*},q_{2*})} $. 
We can use the recently developed techniques of 
\cite{GaravusoKatzarkovNoll} to compute the Hodge diamond of 
$ 
\widetilde{M}^{ 
               (n_{\hat{\eta}_1}, n_{\hat{\gamma}_1}, 
                n_{\hat{\eta}_2}, n_{\hat{\gamma}_2}) 
              }_{(q_{1*},q_{2*})} 
$ 
for each 
$ 
(n_{\hat{\eta}_1}, n_{\hat{\gamma}_1}, 
 n_{\hat{\eta}_2}, n_{\hat{\gamma}_2}) \in B^{(q_{1*},q_{2*})} 
$. 
 
\end{example}

\section{\label{WCP[s_1,...,s_l]}Complete intersection in  
\texorpdfstring{$ \mathbf{WCP}^{m-1|n} $}{WCPm-1|n}} 
 
Consider a member of the family of compact supervarieties  
$  
  \mathbf{WCP}^{m-1|n} 
  _{ ( Q_1,\ldots,Q_m | q_1,\ldots,q_n ) [s_1,\ldots,s_l] } 
$,  
i.e. a complete intersection of hypersurfaces, defined by the common zero locus 
$ \cap_{\beta=1}^l {G_{\beta} = 0} $ of quasihomogeneous polynomials  
$ G_{\beta} = G_{\beta}(\phi,\xi) $ of degree $ s_{\beta} $, in the weighted  
complex superprojective space of Section \ref{WCP}. 
A K\"{a}hler supermanifold may be obtained as a crepant resolution of this  
supervariety if a crepant resolution exists. 
Such a supermanifold would have nonnegative super-first Chern class when 
\begin{equation} 
\label{c_1(WCP[s_1,...,s_l]>=0)} 
\sum_{i=1}^m Q_i - \sum_{a=1}^n q_a -\sum_{\beta = 1}^l s_{\beta} \ge 0 \, . 
\end{equation} 
We will assume that a crepant resolution exists and the condition  
(\ref{c_1(WCP[s_1,...,s_l]>=0)}) is satisfied. 
To describe our K\"{a}hler supermanifold in terms of a gauged linear sigma 
model, we make the replacements 
\begin{equation*}
\overline{P} e^{-2 s V} P 
   \rightarrow \sum_{\beta=1}^l \overline{P}_{\beta} e^{-2 s_{\beta} V} P_{\beta} \, ,
\qquad
P \cdot G(\Phi,\Xi) 
   \rightarrow \sum_{\beta=1}^l P_{\beta} \cdot G_{\beta} \left( \Phi, \Xi \right)    
\end{equation*}
in (\ref{L_WCP[s]}), where the chiral superfield $ P_{\beta} $ has $ U(1) $ charge 
$ - s_{\beta} $  and $ G_{\beta}(\Phi,\Xi) $ is a quasihomogeneous polynomial of $ U(1) $ charge  
$ s_{\beta} $. 
The resulting Lagrangian is 
\begin{align} 
\label{L_WCP[s_1,...,s_l]} 
L &_{ \mathbf{WCP}^{m-1|n} 
      _{ ( Q_1,\ldots,Q_m | q_1,\ldots,q_n ) [s_1,\ldots,s_l] } 
    } 
\nonumber 
\\ 
  &\phantom{\quad}
   = \int d^4 \theta 
     \left( 
             \sum_{i=1}^m \overline{\Phi}_i e^{2 Q_i V} \Phi_i 
           + \sum_{\beta=1}^l \overline{P}_{\beta} e^{-2 s_{\beta} V} P_{\beta}
           + \sum_{a=1}^n \overline{\Xi}_a e^{2 q_a V} \Xi_a 
           - \frac{1}{2 e^2} \overline{\Sigma} \Sigma 
     \right)
\nonumber
\\
  &\phantom{= \quad}      
   - \frac{1}{2} 
     \left( 
            t \int d^2 \tilde{\theta} \,\Sigma + c.c. 
     \right)  
   + \left( 
            \int d^2 \theta \, 
            \sum_{\beta=1}^l 
            P_{\beta} \cdot  G_{\beta} \left( 
                                              \Phi, \Xi 
                                       \right) 
          + c.c. 
     \right). 
\end{align} 
The nonlinear sigma model phase is realized in the low energy limit with 
$ r \gg 0 \, , $ $ \sigma = 0 \, , $ $ p_{\beta} = 0 \, , $ and target space  
\begin{align} 
\overset{l}{\underset{\beta=1}{\bigcap}} 
\left\{ 
        G_{\beta} = 0 
\right\} 
   \subset 
     \frac{ 
            \left\{ ( \phi_1,\ldots,\phi_m | \xi_1,\ldots,\xi_n ) 
               \left| 
                  \textstyle{ 
                              \sum_{i=1}^m Q_i \left| \phi_i \right |^2 
                             + \sum_{a=1}^n q_a \left| \xi_a \right |^2 
                            } 
                             = r 
                    \right. 
            \right\} 
          }{U(1)} \, , 
\end{align} 
where $ p_{\beta} $ is the lowest component of $ P_{\beta} $. 
If the super Calabi-Yau condition 
\begin{equation} 
\label{SCYCI-Condition} 
\sum_{i=1}^m Q_i - \sum_{a=1}^n q_a = \sum_{\beta=1}^l s_{\beta} 
\end{equation} 
is satisfied, then the Fayet-Iliopoulos parameter $ r $ does not renormalize. 
 
Extending the result (\ref{Pi_WCP[s]}) to the present case, we obtain for the 
super Landau-Ginzburg mirror period 
\begin{align} 
\label{Pi_WCP[s_1,...,s_l]} 
\Pi &_{  
        \widetilde{\mathbf{WCP}}^{m-1|n} 
        _{ ( Q_1,\ldots,Q_m | q_1,\ldots,q_n ) [s_1,\ldots s_l] } 
      } 
\nonumber 
\\
   &\phantom{_\mathbf{WCP}} 
    = \int 
      \left( 
            \prod_{i=1}^m dY_i 
      \right) 
      \left( 
             \prod_{\beta=1}^l dY_{P_{\beta}} \, e^{ -Y_{P_{\beta}} } 
      \right) 
      \left( 
            \prod_{a=1}^n dX_a \, d\eta_a \, d \gamma_a 
      \right) 
\nonumber 
\\ 
&\phantom{_\mathbf{WCP}= \int} \times 
\delta \left( 
                \sum_{i=1}^m Q_i Y_i  
              - \sum_{\beta=1}^l s_{\beta} Y_{P_{\beta}} 
              - \sum_{a=1}^n q_a X_a - t 
       \right) 
\nonumber 
\\ 
& \phantom{_\mathbf{WCP}= \int} \times 
  \exp{ \left[ 
              - \sum_{i=1}^m  e^{-Y_i}  
              - \sum_{\beta=1}^l e^{ -Y_{P_{\beta}} } 
              - \sum_{a=1}^n e^{-X_a}(1 + \eta_a \gamma_a) 
        \right] 
      }, 
\end{align} 
where $ Y_{P_\beta} $ is a bosonic twisted chiral superfield (with periodicity 
$ 2 \pi i $) which satisfies 
\begin{equation} 
\mathrm{Re} \, Y_{P_{\beta}}  
  = \overline{P}_{\beta} e^{-2 s_{\beta} V} P_{\beta} \, . 
\end{equation}

\section{\label{T[s_A1,...,s_Al]}Complete intersection in a general toric  
supervariety} 
 
Consider a member of the family of compact supervarieties 
$ \mathcal{T}^{m-k|n}_{(Q_{Ai}|q_{Aa})[s_{A \beta}]} $, i.e. a complete  
intersection of hypersurfaces, defined by the common zero locus 
$ \cap_{\beta=1}^l \{ G_{\beta} = 0 \} $ of quasihomogeneous polynomials  
$ G_{\beta} = G_{\beta}(\phi,\xi) $ of multidegree $ s_{A \beta} $, where  
$ A=1,\ldots,k $, in the toric supervariety   
\begin{equation*} 
\mathcal{T}^{m-k|n}_{(Q_{Ai}|q_{Aa})} 
   = \frac{ \mathbf{C}^{m|n} \setminus S }{ ( \mathbf{C}^{*} )^k } \, . 
\end{equation*} 
Here, the matrix elements of the matrices 
\begin{equation*} 
(Q_{Ai} | q_{Aa} ) =  
\left( 
\begin{array}{ccc|ccc} 
Q_{11} & \cdots & Q_{1m} & q_{11} & \cdots & q_{1n} 
\\ 
\vdots & & \vdots & \vdots & & \vdots 
\\ 
Q_{k1} & \cdots & Q_{km} & q_{k1} & \cdots & q_{kn} 
\end{array} 
\right), 
\qquad 
[s_{A \beta}] = 
\begin{pmatrix} 
s_{11} & \cdots & s_{1l} 
\\ 
\vdots & & \vdots 
\\ 
s_{k1} & \cdots & s_{kl} 
\end{pmatrix} 
\end{equation*} 
are positive integers, $ S $ is an exceptional supervariety in  
$ \mathbf{C}^{m|n} $, and  
\begin{equation*} 
(\mathbf{C}^*)^k: \  
\phi_i \rightarrow \lambda^{ Q_{A i} } \phi_i \, , 
\quad 
\xi_a \rightarrow \lambda^{ q_{A a} } \xi_a \, , 
\qquad 
\begin{array}{ll} 
i = 1,\ldots,m \, ; \quad 
&  
a = 1,\ldots,n \, ; 
\\ 
A = 1,\ldots,k \, ; \quad 
& 
\lambda \in \mathbf{C}^* \, . 
\end{array} 
\end{equation*} 
A K\"ahler supermanifold may be obtained as a crepant resolution of this  
supervariety if a crepant resolution exists. 
Such a supermanifold would have nonnegative super-first Chern class when 
\begin{equation} 
\label{c_1([s_A1,...,s_Al])>=0} 
\sum_{i=1}^m Q_{Ai} - \sum_{a=1}^n q_{Aa} 
   \ge \sum_{\beta=1}^l s_{A \beta} \, , 
\qquad 
A = 1,\ldots,k \, .  
\end{equation} 
We will assume that a crepant resolution exists and the condition  
(\ref{c_1([s_A1,...,s_Al])>=0}) is satisfied. 
To describe our K\"ahler supermanifold in terms of of a gauged linear sigma  
model, we replace the $ U(1) $ gauge group of  
(\ref{L_WCP[s_1,...,s_l]}) with $ U(1)^k $.  
The resulting Lagrangian is  
\begin{align} 
L &_{  \mathcal{T}^{m-k|n}_{ ( Q_{Ai}|q_{Aa} ) [s_{A \beta}] }  } 
\nonumber 
\\ 
  &\phantom{\qquad}
   = \int d^4 \theta 
     \left( 
             \sum_{i=1}^m \overline{\Phi}_i  
             e^{ 2 \sum_{A=1}^k Q_{Ai} V_{A} } \Phi_i
           + \sum_{\beta=1}^l \overline{P}_{\beta} 
             e^{-2 \sum_{A=1}^k s_{A \beta} V_A} P_{\beta}
     \right. 
\nonumber
\\
  &\phantom{\qquad = \int d^4 \theta}
     \left.                
           + \sum_{a=1}^n \overline{\Xi}_a  
             e^{ 2 \sum_{A=1}^k q_{Aa} V_{A} } \Xi_a                     
           - \sum_{A,B = 1}^k 
             \frac{1}{ 2 e_{AB}^2 }  
             \overline{\Sigma}_A \Sigma_B 
    \right) 
\nonumber 
\\ 
  &\phantom{\qquad=} 
   - \frac{1}{2} 
     \left( 
            \int d^2 \tilde{\theta} \,  
            \sum_{A=1}^k t_A \Sigma_A  
           + c.c. 
     \right) 
   + \left( 
            \int d^2 \theta \, 
            \sum_{\beta=1}^l 
            P_{\beta} \cdot G_{\beta} \left( 
                                             \Phi, \Xi 
                                      \right) 
          + c.c. 
     \right). 
\end{align} 
Under the $ A $-th $ U(1) $, the bosonic chiral superfields 
$ \Phi_i $ have charges $ Q_{Ai} $, where $ i=1,\ldots,m $, 
the fermionic chiral superfields 
$ \Xi_a $ have charges $ q_{Aa} $, where $ a=1,\ldots,n $, 
the bosonic chiral superfields $ P_{\beta} $ 
have charges $ -s_{A \beta} $, and the quasihomogeneous  
polynomials $ G_{\beta} = G_{\beta} ( \Phi, \Xi ) $ have charges  
$ s_{A \beta} $, where $ \beta = 1,\ldots,l $. 
The nonlinear sigma model phase is realized in the low energy limit with 
$ r_A \gg 0 $, $ \sigma_A = 0 $,  $ p_{\beta} = 0 $, and target space 
\begin{equation} 
     \overset{l}{\underset{\beta=1}{\bigcap}} 
     \left\{ 
            G_{\beta} = 0 
     \right\} 
   \subset 
     \frac{ 
            \left\{ ( \phi_1,\ldots,\phi_m | \xi_1,\ldots,\xi_n ) 
               \left| 
                  \textstyle{ 
                              \sum_{i=1}^m Q_{Ai}  
                              \left| \phi_i \right |^2 
                            + \sum_{a=1}^n q_{Aa}  
                              \left| \xi_a \right |^2 
                            } 
                             = r_A 
                    \right. 
            \right\} 
          } 
{U(1)^k} \, , 
\end{equation} 
where $ \sigma_{A} $ is the lowest component of $ \Sigma_{A} $. 
If the super Calabi-Yau condition 
\begin{equation} 
\label{SCYCIT-Condition} 
\sum_{i=1}^m Q_{Ai} - \sum_{a=1}^n q_{Aa} 
   = \sum_{\beta=1}^l s_{A \beta} \, , 
\qquad 
A = 1,\ldots,k 
\end{equation} 
is satisfied, then the $ r_A $ do not renormalize. 
 
Extending the result (\ref{Pi_WCP[s_1,...,s_l]}) to the present case, we  
obtain the for super Landau-Ginzburg mirror period 
\begin{align} 
\label{Pi_T[s_A1,...,s_Al]} 
\Pi_{  \widetilde{\mathcal{T}}^{m-k|n}_{ ( Q_{Ai}|q_{Aa} ) [s_{A \beta}] }  } 
   &= \int 
      \left( 
            \prod_{i=1}^m dY_i 
      \right) 
      \left( 
             \prod_{\beta=1}^l dY_{P_{\beta}} \, e^{ -Y_{P_{\beta}} } 
      \right) 
      \left( 
            \prod_{a=1}^n dX_a \, d\eta_a \, d \gamma_a 
      \right) 
\nonumber 
\\ 
&\phantom{= \int} \times 
\prod_{A=1}^k 
\delta \left( 
                \sum_{i=1}^m Q_{Ai} Y_i  
              - \sum_{\beta=1}^l s_{A \beta} Y_{P_{\beta}} 
              - \sum_{a=1}^n q_{Aa} X_a - t_A 
       \right) 
\nonumber 
\\ 
& \phantom{= \int}\times 
  \exp{ \left[ 
              - \sum_{i=1}^m  e^{-Y_i}  
              - \sum_{\beta=1}^l e^{ -Y_{P_{\beta}} } 
              - \sum_{a=1}^n e^{-X_a}(1 + \eta_a \gamma_a) 
        \right] 
      }, 
\end{align} 
where 
\begin{gather} 
\mathrm{Re} \, Y_i  
  = \overline{\Phi}_i  
    \exp{  
          \left(  
                 2 \sum_{A=1}^k Q_{Ai} V_A 
          \right)  
         }  
    \Phi_i \, , 
\qquad 
\mathrm{Re} \, Y_{P_{\beta}}  
  = \overline{P}_{\beta}  
    \exp{ \left(  
                 -2 \sum_{A=1}^k s_{A \beta} V_A 
          \right) 
        }  
    P_{\beta}  
\, , 
\nonumber 
\\ 
\mathrm{Re} \, X_a  
  = - \overline{\Xi}_a  
      \exp{ \left(  
                   2 \sum_{A=1}^k q_{Aa} V_A 
            \right)   
          }  
      \Xi_a \, . 
\end{gather} 
 
Consider making the change of variables 
\begin{equation} 
e^{-Y_{P_{\beta}}} = \widetilde{P}_{\beta} \, , 
\qquad 
e^{-Y_i} = U_i \prod_{\beta=1}^l \widetilde{P}_\beta^{M_{\beta i}} \, , 
\qquad 
e^{-X_a} = V_a \prod_{\beta=1}^l \widetilde{P}_\beta^{N_{\beta a}} \, , 
\end{equation} 
in (\ref{Pi_T[s_A1,...,s_Al]}), where the matrices  
$ (M_{\beta i}) $ and $ (N_{\beta a}) $ satisfy 
\begin{equation} 
\label{s_beta} 
s_{A \beta}  
   = \sum_{i=1}^m M_{\beta i} Q_{Ai} - \sum_{a=1}^n N_{\beta a} q_{Aa} \, , 
\qquad 
A = 1,\ldots,k \, ; 
\quad 
\beta = 1,\ldots,l 
\end{equation} 
and, for fixed $ i = \hat{\imath} $, $ a = \bar{a} $,   
have \emph{at most} one nonzero matrix element  
$ M_{ \hat{\beta} \hat{\imath} } = 1 $,  
$ N_{ \bar{\beta} \bar{a} } = 1 $, respectively. 
Then, in terms of the new variables, we obtain  
\begin{align*} 
\Pi &_{    
        \widetilde{\mathcal{T}}^{m-k|n}_{ ( Q_{Ai} | q_{Aa} ) [s_{A \beta}] } 
      } 
\nonumber 
\\ 
   &= (-1)^{m+n+l} \int 
      \left( 
             \prod_{i=1}^m \frac{dU_i}{U_i} 
      \right) 
      \left( 
             \prod_{\beta=1}^l d\widetilde{P}_{\beta} 
      \right) 
      \left( 
             \prod_{a=1}^n \frac{dV_a}{V_a} \,  
             d\eta_a \, d \gamma_a 
      \right) 
\nonumber 
\\ 
&\phantom{=(-1)^{m+n+l} \int} 
\times \prod_{A=1}^k 
\delta \left( 
              \ln{ \left( 
                          \frac{ \prod_{i=1}^m U^{Q_{Ai}}_i } 
                               { \prod_{a=1}^n V^{q_{Aa}}_a } 
                   \right)       
                 }  
            + t_A 
       \right) 
\nonumber 
\\ 
&\phantom{=(-1)^{m+n+l} \int} \times 
  \exp{ \left[ 
              - \sum_{\beta=1}^l \widetilde{P}_{\beta} 
                \left( 
                       \sum_{ M_{\beta i} = 1 } U_i 
                     + \sum_{ N_{\beta a} = 1 } V_a  
                       \left( 
                              1 + \eta_a \gamma_a 
                       \right)  
                     + 1 
                \right) 
         \right. 
      } 
\nonumber 
\\ 
&\phantom{=(-1)^{m+n+l} \int} 
 \phantom{ 
           \times  
           \exp{  
                 \left[  
                        \vphantom{ \sum_{\beta=1}^l  
                                   \widetilde{P}_{\beta} 
                                 }  
                  \right.  
               }  
         } 
        \left. 
              - \sum_{ M_{\beta i} = 0 \ \forall \beta } U_i 
              - \sum_{ N_{\beta a} = 0 \ \forall \beta } V_a 
                \left( 
                       1 + \eta_a \gamma_a 
                \right)        
        \right]. 
\end{align*} 
Performing the integration over the $ \widetilde{P}_{\beta} $ yields 
\begin{align} 
\label{MirrorPeriodSuperCI} 
\Pi &_{   
        \widetilde{\mathcal{T}}^{m-k|n}_{ ( Q_{Ai} | q_{Aa} ) [s_{A \beta}] } 
      } 
\nonumber 
\\ 
  &= (-1)^{m+n+l} \int 
      \left( 
             \prod_{i=1}^m \frac{dU_i}{U_i} 
      \right) 
      \left( 
             \prod_{a=1}^n \frac{dV_a}{V_a} \, 
             d\eta_a \, d \gamma_a 
      \right) 
\prod_{A=1}^k 
\delta \left( 
              \ln{ \left( 
                          \frac{ \prod_{i=1}^m U^{Q_{Ai}}_i } 
                               { \prod_{a=1}^n V^{q_{Aa}}_a } 
                   \right) 
                 } 
            + t_A 
       \right) 
\nonumber 
\\ 
&\phantom{=(-1)^{m+n+l} \int} 
  \times \prod_{\beta=1}^l  
  \delta \left( 
                \sum_{ M_{\beta i} = 1 } U_i 
              + \sum_{ N_{\beta a} = 1 } V_a 
                \left( 
                       1 + \eta_a \gamma_a 
                \right) 
              + 1 
         \right) 
\nonumber 
\\ 
&\phantom{=(-1)^{m+n+l} \int} \times 
         \exp{ 
               \left( 
                    - \sum_{ M_{\beta i} = 0 \ \forall \beta } U_i 
                    - \sum_{ N_{\beta a} = 0 \ \forall \beta } V_a 
                      \left( 
                             1 + \eta_a \gamma_a 
                      \right)    
               \right) 
             }. 
\end{align} 
Thus, we have obtained an $ (m-k-l-n) $-dimensional noncompact supermanifold  
$ \widetilde{M}^{\circ} \subset (\mathbf{C}^*)^{m + n|2n} $ defined by 
\begin{equation} 
\prod_{i=1}^m U^{Q_{Ai}}_i  \prod_{a=1}^n V^{-q_{Aa}}_a = e^{-t_A} 
\, , 
\qquad 
  \sum_{ M_{\beta i} = 1 } U_i 
+ \sum_{ N_{\beta a} = 1 } V_a 
  \left( 
         1 + \eta_a \gamma_a 
  \right) 
+ 1 
= 0 \, .  
\end{equation} 
The period (\ref{MirrorPeriodSuperCI}) is identical to the period of a  
super Landau-Ginzburg model on $ \widetilde{M}^{\circ} $ with superpotential 
\begin{equation} 
\widetilde{W}_{\widetilde{M}^{\circ}}  
  =   \sum_{ M_{\beta i} = 0 \ \forall \beta } U_i 
    + \sum_{ N_{\beta a} = 0 \ \forall \beta } V_a 
      \left( 
             1 + \eta_a \gamma_a 
      \right). 
\end{equation} 
 
\begin{example} 

Consider a gauged linear sigma model corresponding to the complete  
intersection defined by two quadrics in $ \mathbf{CP}^{3|2} $. 
We must find matrices $ (M_{\beta i}) $ and $ (N_{\beta a}) $ which satisfy 
(\ref{s_beta}) and, for fixed $ i = \hat{\imath} $, $ a = \bar{a} $, 
have at most one nonzero matrix element 
$ M_{ \hat{\beta} \hat{\imath} } = 1 $, 
$ N_{ \bar{\beta} \bar{a} } = 1 $, respectively. 
We can choose 
\begin{equation*} 
(M_{\beta i})  
  =  \begin{pmatrix} 
        1 & 1 & 0 & 0 \\ 
        0 & 0 & 1 & 1 
     \end{pmatrix}, 
\qquad 
(N_{\beta a}) 
  = \begin{pmatrix} 
         0 & 0 \\ 
         0 & 0  
     \end{pmatrix}. 
\end{equation*} 
The mirror theory is thus a super Landau-Ginzburg model on the supermanifold  
$ \widetilde{M}^{\circ} \subset (\mathbf{C}^*)^{(6|4)}$  
defined by 
\begin{equation*} 
U_1 U_2 U_3 U_4 V^{-1}_1 V^{-1}_2 = e^{-t} \, , 
\qquad 
U_1 + U_2 + 1 = 0 \, , 
\qquad 
U_3 + U_4 + 1 = 0  
\end{equation*} 
and with superpotential 
\begin{equation*} 
\widetilde{W}_{\widetilde{M}^{\circ}}  
  = V_1 \left(  
               1 + \eta_1 \gamma_1  
        \right)  
  + V_2 \left( 1 + \eta_2 \gamma_2  
        \right).  
\end{equation*}
 
\end{example}

\section{\label{PeriodRelations}Period relations} 
 
In this section, we will establish some relations between periods of mirrors  
of gauged linear sigma models corresponding to various geometries. 
These relations generalize results which were obtained in \cite{Schwarz:Sigma} 
by working with A-models. 
 
Starting from (\ref{Pi_WCP[s]}), we obtain 
\begin{align*} 
\Pi &{}_{  
          \widetilde{\mathbf{WCP}}^{m-1|n} 
          _{ ( Q_1,\ldots,Q_m | q_1,\ldots,q_n ) [s] } 
        } 
\nonumber 
\\[11pt] 
 &=   \int 
      \left( 
            \prod_{i=1}^m dY_i 
      \right) 
      dY_P \, e^{-Y_P}   
      \left( 
            \prod_{a=1}^n dX_a \, d\eta_a \, d \gamma_a 
      \right) 
\delta \left( 
                \sum_{i=1}^m Q_i Y_i - s Y_P 
              - \sum_{a=1}^n q_a X_a - t 
       \right)
\nonumber 
\\ 
& \phantom{=\int} 
  \times 
  \exp{ \left[ 
              - \sum_{i=1}^m  e^{-Y_i} - e^{-Y_P} 
              - \sum_{a=1}^n e^{-X_a}(1 + \eta_a \gamma_a) 
        \right] 
      }
\allowdisplaybreaks       
\nonumber 
\\[11pt] 
   &= (-1)^n \int 
      \left( 
            \prod_{i=1}^m dY_i 
      \right) 
      dY_P \, e^{-Y_P} 
      \left( 
            \prod_{a=1}^n dX_a \, e^{-X_a} 
      \right) 
\nonumber 
\delta \left( 
                \sum_{i=1}^m Q_i Y_i - s Y_P 
              - \sum_{a=1}^n q_a X_a - t 
       \right) 
\nonumber 
\\ 
&\phantom{=(-1)^n \int} \times 
  \exp{ \left( 
              - \sum_{i=1}^m  e^{-Y_i} -  e^{-Y_P} 
              - \sum_{a=1}^n e^{-X_a} 
        \right) 
      } 
\allowdisplaybreaks
\nonumber 
\\[11pt] 
&= 
\Pi_{  
      \widetilde{\mathbf{WCP}}^{m-1} 
      _{ ( Q_1,\ldots,Q_m ) [s,q_1,\ldots,q_n] } 
    } \, . 
\end{align*} 
Thus, 
\begin{equation} 
\label{Relation1} 
\Pi_{ 
      \widetilde{\mathbf{WCP}}^{m-1|n} 
      _{ ( Q_1,\ldots,Q_m | q_1,\ldots,q_n )[s] } 
    } 
= 
\Pi_{ 
      \widetilde{\mathbf{WCP}}^{m-1} 
      _{ ( Q_1,\ldots,Q_m ) [s,q_1,\ldots,q_n] } 
    } 
\, . 
\end{equation} 
Similarly, starting from (\ref{Pi_WCP[s_1,...,s_l]}), we obtain 
\begin{align*} 
\Pi &{}_{ 
          \widetilde{\mathbf{WCP}}^{m-1|n} 
          _{ ( Q_1,\ldots,Q_m | q_1,\ldots,q_n ) [s_1,\ldots,s_l] } 
         } 
\nonumber 
\\[11pt] 
   &= \int 
      \left( 
            \prod_{i=1}^m dY_i 
      \right) 
      \left( 
             \prod_{\beta=1}^l dY_{P_{\beta}} \, e^{ -Y_{P_{\beta}} } 
      \right) 
      \left( 
            \prod_{a=1}^n dX_a \, d\eta_a \, d \gamma_a 
      \right) 
\nonumber 
\\ 
&\phantom{=\int} \times 
\delta \left( 
                \sum_{i=1}^m Q_i Y_i 
              - \sum_{\beta=1}^l s_{\beta} Y_{P_{\beta}} 
              - \sum_{a=1}^n q_a X_a - t 
       \right) 
\nonumber 
\\ 
&\phantom{=\int} \times 
  \exp{ \left[ 
              - \sum_{i=1}^m  e^{-Y_i} 
              - \sum_{\beta=1}^l e^{ -Y_{P_{\beta}} } 
              - \sum_{a=1}^n e^{-X_a}(1 + \eta_a \gamma_a) 
        \right] 
      } 
\allowdisplaybreaks 
\nonumber 
\\[11pt] 
   &= (-1)^{l-1}\int 
      \left( 
            \prod_{i=1}^m dY_i 
      \right) 
      dY_{P_l} \, e^{-Y_{P_l}} 
      \left( 
             \prod_{\beta=1}^{l-1} dY_{P_{\beta}} \,  
             \eta_{P_{\beta}} \gamma_{P_{\beta}}  
      \right) 
      \left( 
            \prod_{a=1}^n dX_a \, d\eta_a \, d \gamma_a 
      \right) 
\nonumber 
\\ 
&\phantom{=(-1)^{l-1}\int} \times 
\delta \left( 
                \sum_{i=1}^m Q_i Y_i 
              - s_l Y_{P_l}  
              - \sum_{\beta=1}^{l-1} s_{\beta} Y_{P_{\beta}} 
              - \sum_{a=1}^n q_a X_a - t 
       \right) 
\nonumber 
\\ 
&\phantom{=(-1)^{l-1}\int} \times 
  \exp{ \left[ 
              - \sum_{i=1}^m  e^{-Y_i} 
              - e^{ -Y_{P_l} } 
              - \sum_{\beta=1}^{l-1} e^{ -Y_{P_{\beta}} } 
                (1 + \eta_{P_{\beta}} \gamma_{P_{\beta}}) 
              - \sum_{a=1}^n e^{-X_a}(1 + \eta_a \gamma_a) 
        \right] 
       }  
\allowdisplaybreaks 
\nonumber 
\\[11pt] 
&= 
   \Pi_{ 
         \widetilde{\mathbf{WCP}}^{m-1|n+l-1} 
         _{ ( Q_1,\ldots,Q_m | q_1,\ldots,q_n, s_1,\ldots,s_{l-1} ) [s_l] }  
       } 
\allowdisplaybreaks
\nonumber 
\\[11pt] 
   &= (-1)^l\int 
      \left( 
            \prod_{i=1}^m dY_i 
      \right) 
      \left( 
             \prod_{\beta=1}^l dY_{P_{\beta}} \, 
             \eta_{P_{\beta}} \gamma_{P_{\beta}} 
      \right) 
      \left( 
            \prod_{a=1}^n dX_a \, d\eta_a \, d \gamma_a 
      \right) 
\nonumber 
\\ 
&\phantom{=(-1)^l\int} \times 
\delta \left( 
                \sum_{i=1}^m Q_i Y_i 
              - \sum_{\beta=1}^l s_{\beta} Y_{P_{\beta}} 
              - \sum_{a=1}^n q_a X_a - t 
       \right) 
\nonumber 
\\ 
&\phantom{=(-1)^l\int} \times 
  \exp{ \left[ 
              - \sum_{i=1}^m  e^{-Y_i} 
              - \sum_{\beta=1}^l e^{ -Y_{P_{\beta}} } 
                (1 + \eta_{P_{\beta}} \gamma_{P_{\beta}}) 
              - \sum_{a=1}^n e^{-X_a}(1 + \eta_a \gamma_a) 
        \right] 
      } 
\nonumber 
\\[11pt] 
&= 
   \Pi_{ 
         \widetilde{\mathbf{WCP}}^{m-1|n+l} 
         _{ ( Q_1,\ldots,Q_m | q_1,\ldots,q_n, s_1,\ldots,s_l ) }  
       } \, . 
\end{align*} 
Thus, 
\begin{align} 
\Pi_{ 
      \widetilde{\mathbf{WCP}}^{m-1|n} 
      _{ ( Q_1,\ldots,Q_m | q_1,\ldots,q_n ) [s_1,\ldots,s_l] } 
    } 
&= 
\Pi_{ 
      \widetilde{\mathbf{WCP}}^{m-1|n+l-1} 
      _{ ( Q_1,\ldots,Q_m | q_1,\ldots,q_n, s_1,\ldots,s_{l-1} ) [s_l] } 
    } 
\nonumber 
\\[11pt] 
&= 
\label{Relation2} 
\Pi_{ 
      \widetilde{\mathbf{WCP}}^{m-1|n+l} 
      _{ ( Q_1,\ldots,Q_m | q_1,\ldots,q_n, s_1,\ldots,s_l ) } 
    } \, . 
\end{align} 
In a straightforward manner, one can generalize the relation 
(\ref{Relation1}) to obtain 
\begin{equation} 
\Pi_{  \widetilde{\mathcal{T}}^{m-k|n}_{ ( Q_{Ai}|q_{Aa} ) [s_A] }  } 
  = \Pi_{ \widetilde{\mathcal{T}}^{m-k|n}_{ ( Q_{Ai} ) [s_A,q_{Aa}] }  } 
\end{equation} 
and generalize (\ref{Relation2}) to obtain 
\begin{equation} 
\Pi_{  \widetilde{\mathcal{T}}^{m-k|n} 
       _{ ( Q_{Ai}|q_{Aa} ) [s_{A1},\ldots,s_{Al}] }  } 
  = \Pi_{  \widetilde{\mathcal{T}}^{m-k|n} 
    _{ ( Q_{Ai} | q_{Aa},s_{A1},\ldots,s_{A \, l-1}) [s_{Al}] }  } 
  = \Pi_{  \widetilde{\mathcal{T}}^{m-k|n} 
    _{ ( Q_{Ai} | q_{Aa},s_{A1},\ldots,s_{Al} ) }  } 
\, . 
\end{equation}

\section{\label{Categorical}Categorical interpretation}  
In this section, we give a categorical intepretation of the material presented in 
Sections \ref{WCP} - \ref{PeriodRelations}. 
We briefly review the phenomena of spectra and gaps \cite{BFK} and connect them with supercohomology calculations and algebraic cycles. 
Complete details will appear in \cite{GaravusoKatzarkovNoll}. 
 
\subsection{Supermanifold cohomology calculations} 
 
We begin with a simple example. 
 
\begin{example} 
\label{super3Dcubic} 
Consider a hypersurface 
 
\begin{equation*} 
M = \left\{ 
            \phi_1^3 + \phi_2^3 + \phi_3^3 + \phi_4^3 + \phi_5^3  
          + \xi_1 \xi_2  
          = 0 
    \right\} 
  \in \mathbf{WCP}^{4|2}_{(1,1,1,1,1|1,2)[3]} \, .  
\end{equation*} 
Treating $ M $ as a DG scheme, its structure sheaf becomes 
the complex shown below.  
 
\begin{equation*}  
\xymatrix{ 
\cO(-4)\ar[r]\ar[d] 
& 
\cO(-5)^{\oplus 2}\ar[r]\ar[d] 
& 
\cO(-1)\ar[d] 
\\ 
\cO(-2)\oplus\cO(-2)\ar[r]\ar[d] 
& 
\cO(-3)^{\oplus 4}\ar[r]\ar[d]\ar[dl] 
& 
\cO(-4)^{\oplus 2}\ar[d]\ar[dl] 
\\ 
\cO\ar[r]^0 
& 
\cO(-1)^{\oplus 2}\ar[r]^0 
& 
\cO(-2) 
} 
\end{equation*} 
Note that 
\begin{equation*} 
h^1(\cO) = h^1 \left( K_{\mathbf{CP}^5} - 2H - 2H - 2H \right) 
         = 1 \, . 
\end{equation*} 
Here $ H $ is the  hyperplane section in $ \mathbf{CP}^5 $. 
Applying the Riemman-Roch theorem, we obtain 
\begin{equation*} 
h^{2,1}(M) = 5 \, , 
\qquad 
h^{1,1}(M) = 1 \, . 
\end{equation*} 
\end{example} 
Example \ref{super3Dcubic} gives us a flavor of supermanifold cohomology  
calculations. 
In general, each supermanifold Hodge number is the sum of a number obtained 
from a DG scheme calculation and a number coming from the gaps of the category of singularities; see Table 
\ref{tab:HC31} and Definition \ref{OrlovSpec}. 
Further details can be found in \cite{GaravusoKatzarkovNoll}. 
 
On the symplectic side, the Example \ref{super3Dcubic} calculation can be seen tropically as
\cite{KRIS}:  
\begin{center}  
\includegraphics{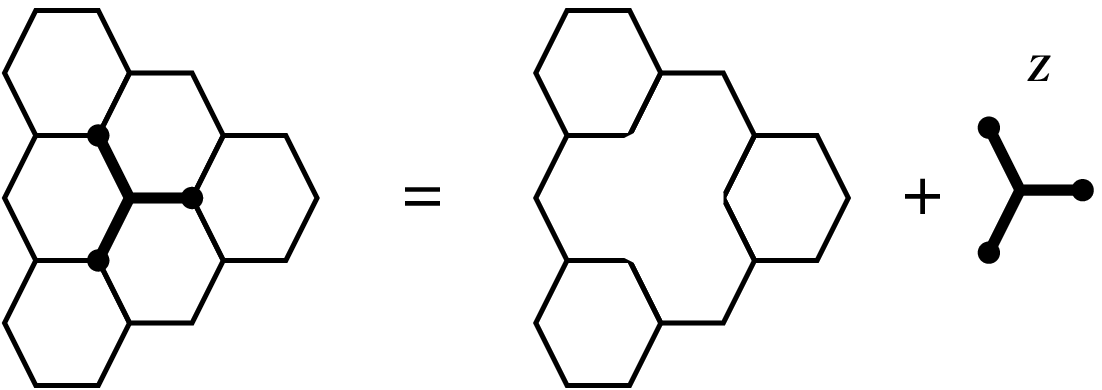} 
\end{center}
Here, $ Z $ is the mirror of $ \mathbf{WCP}^{1|2} $. 
We thus have the A-side categorical equivalence 
\begin{equation*}  
\Fuk(M)^d 
   = \left< \Fuk(\text{genus 4 curve}), \Fuk(Z),... \right>,  
\end{equation*} 
where $ \Fuk $ denotes the Fukaya category. 
In other words, the B-side noncommutative deformation corresponds to a conifold transition, which 
produces the Fukaya category of a genus $ 4 $ curve and the mirror of $ \mathbf{WCP}^{1|2} $. 
This is a special case of Theorem \ref{decompositionA-side}.

\subsection{Super and DG scheme version of Homological Mirror Symmetry}  
 
In what follows, $ D^b(A) $ is the bounded derived category of coherent sheaves on the scheme $ A $, $ D^b(A)^d $ is a (possibly trivial) deformation of $ D^b(A) $, and $ \langle D^b(A),\ldots \rangle $ is a semiorthogonal decomposition which includes $ D^b(A) $ as a summand. 
Similarly, $ \mathsf{Fuk}(A)^d $ is a (possibly trivial) deformation of  
$ \mathsf{Fuk}(A) $ and $ \langle \mathsf{Fuk}(A),\ldots \rangle $ is a semiorthogonal decomposition which includes $ \mathsf{Fuk}(A) $ as a summand. 
 
\begin{theorem} 
\label{decomposition} 
\begin{equation*}  
D^b \left( 
           \mathbf{WCP}^{m-1|n}_{ (Q_1, \ldots, Q_m|q_1, \ldots, q_n) } 
    \right)^d  
   = \left< 
            D^b \left( 
                       \mathbf{WCP}^{m-1} 
                       _{ (Q_1, \dots, Q_m)[q_1, \ldots, q_n] } 
                \right), \ldots  
     \right>. 
\end{equation*} 
\end{theorem} 
The proof follows from \cite{Orlov} and is based on the identity 
\begin{equation*} 
\mathbf{C}[\phi_1, \ldots \phi_m, \xi_1, \ldots, \xi_n]_{((u))}  
   = \mathbf{C}[\phi_{k+1}, \dots, \phi_m]_{((u))}  
   \otimes \prod_{j=1}^k \mathbf{C}[\phi_j, \xi_1, \ldots, \xi_n] \, . 
\end{equation*} 
Theorem \ref{decomposition} gives a new proof of a theorem due to Bardelli and  
M\"{u}ller-Stach \cite{BMS}. 
We can restate Theorem \ref{decomposition} by saying that some noncommutative deformation of  
$ D^b \left( \mathbf{WCP}^{m-1|n}_{ (Q_1, \ldots, Q_m|q_1, \ldots, q_n) } \right) $ yields  
$ D^b \left( \mathbf{WCP}^{m-1}_{(Q_1, \ldots, Q_m)[q_1, \ldots, q_n]} \right) $.  
Alternatively, we can say that 
$ D^b \left( \mathbf{WCP}^{m-1|n}_{ (Q_1, \ldots, Q_m|q_1, \ldots, q_n) } \right) $ contains as a  
semiorthogonal summand a noncommutative Calabi-Yau (or Fano or general type manifold).  
The following theorem is the A-side version of this statement. 
\begin{theorem}  
\begin{equation*}
\label{decompositionA-side}   
\Fuk \left( \mathbf{WCP}^{m-1|n}_{ (Q_1, \ldots, Q_n | q_1, \ldots, q_n) }  
     \right)^d  
   = \left< 
            \Fuk  
            \left(  
                   \mathbf{WCP}^{m-1}_{ (Q_1, \ldots, Q_m)[q_1, \dots, q_n] } 
            \right),  
            \ldots 
     \right>. 
\end{equation*} 
\end{theorem} 
Figure \ref{ClassicalHMS} gives a schematic picture of classical Homological Mirror Symmetry in the version which is relevant for our purpose.  
For more details, see \cite{KRIS}. 
\begin{figure} 
\begin{center}  
\begin{tabular}[c]{c|c} \hline  
A-models (symplectic) & B-models (algebraic)  
\\ \hline \\  
$ X = (X,\omega) $ a closed symplectic manifold &  
$ X $ a smooth projective variety  
\\[1.2em]  
\begin{minipage}[t]{0.45\linewidth}  
\emph{Fukaya category} $ \Fuk(X)$: 
Objects are Lagrangian submanifolds $L$ which may be equipped with flat line  
bundles.  
Morphisms are given by Floer cohomology $ HF^*(L_0,L_1) $.  
\end{minipage} &  
\begin{minipage}[t]{0.45\linewidth}  
\emph{Derived category} $ D^b(X) $:  
Objects are complexes of coherent sheaves $ \mathcal{E} $.  
Morphisms are $ Ext^*({\mathcal E}_0,{\mathcal E}_1) $.  
\end{minipage}  
\\  
\multicolumn{2}{c}{\xymatrix{ \ar@{<->}[drr]&&\ar@{<->}[dll]\\ && } }  
\medskip  
\\  
\begin{minipage}[t]{0.45\linewidth}  
$Y$ a non-compact symplectic manifold with a proper map  
$ W:Y \to \mathbf{C} $ which is a symplectic fibration with singularities.  
\end{minipage}\ & \  
\begin{minipage}[t]{0.45\linewidth}  
$ Y $ a smooth quasi-projective variety with a proper holomorphic map  
$ W: Y \to \mathbf{C} $.  
\end{minipage}  
\\[3.5em] 
\\  
\begin{minipage}[t]{0.45\linewidth}  
\emph{Fukaya-Seidel category of the Landau-Ginzburg model} $ FS(LG)$:  
Objects are Lagrangian submanifolds $ L \subset Y $ which, at infinity, are  
fibered over $ \mathbf{R}^+ \subset \mathbf{C} $. 
The morphisms are $ HF^*(L_0^+,L_1) $, where the superscript  
$ + $ indicates a perturbation removing intersection points at infinity. 
\end{minipage} &  
\begin{minipage}[t]{0.45\linewidth}  
The category $ D^b_{sing}(W) $ of algebraic $ B $-branes which is obtained by  
considering the singular fibers $ Y_z = W^{-1}(z) $, dividing  
$ D^b(Y_z) $ by the subcategory of perfect complexes {\it Perf}$ \,(Y_z) $,  
and then taking the direct sum over all such $ z $.  
\end{minipage}  
\end{tabular}  
\end{center} 
\caption{Classical Homological Mirror Symmetry} 
\label{ClassicalHMS} 
\end{figure} 
Theorems \ref{decomposition} and \ref{decompositionA-side} can be seen as an extension of Homological Mirror Symmetry to the case of DG schemes. 
More details will appear in \cite{GaravusoKatzarkovNoll}. 
 
\begin{rem} 
Theorem \ref{decompositionA-side} can be seen as a symplectic version of Orlov's theorem \cite{Orlov}.  
\end{rem} 
 
\begin{rem} 
Theorem \ref{decompositionA-side} can be extended to hypersurfaces and complete intersections
in general toric supervarieties.
\end{rem} 
Our categorical findings are summarized in Figures \ref{fig:genpic} and \ref{fig:table_b}. 
These figures give a natural super generalization of classical Homological Mirror Symmetry. 
 
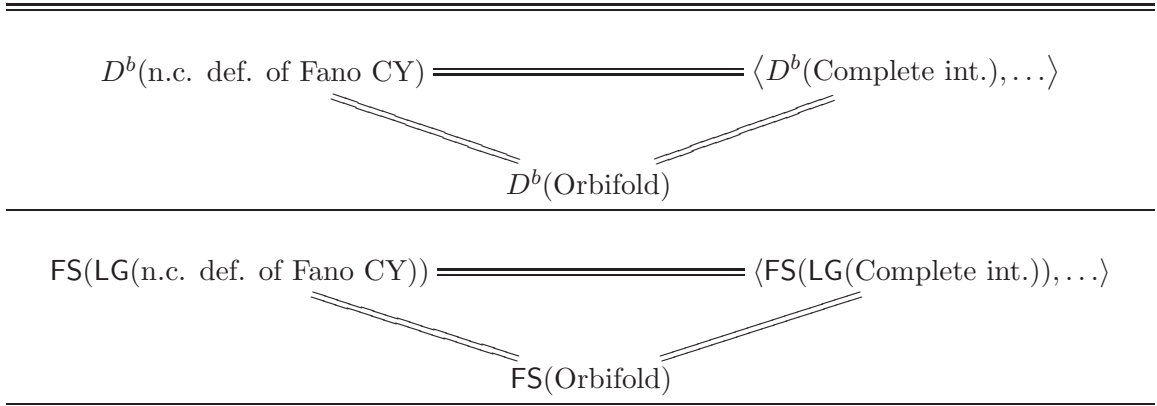
\begin{figure}[htbp] 
  \centering 
  \begin{tabular}{c} 
    \hline\hline 
    \begin{minipage}[c]{.93\linewidth} 
      \[\xymatrix{ 
        D^b(\text{n.c. def. of Fano CY})\ar@{=}[rr]\ar@{=}[dr] 
        && 
         \left< D^b(\text{Complete int.}), \ldots  
\right>\ar@{=}[dl] 
        \\ 
        & 
        D^b(\text{Orbifold}) } 
      \] 
    \end{minipage} 
    \\\hline 
    \begin{minipage}[c]{.97\linewidth} 
      \begin{equation*} 
        \xymatrix{ 
        \FS(\LG(\text{n.c. def. of Fano CY}))\ar@{=}[rr]\ar@{=}[dr] 
        && 
        \left< \FS(\LG(\text{Complete int.})), \ldots \right>\ar@{=}[dl] 
        \\ 
        & 
        \FS(\text{Orbifold}) } 
       \end{equation*} 
    \end{minipage} 
    \\\hline\hline 
  \end{tabular} 
  \caption{General picture of noncommutative deformations } 
  \label{fig:genpic} 
\end{figure} 
 
\begin{figure}[htbp] 
  \centering 
  \begin{tabular}{p{0.41\textwidth}|p{0.47\textwidth}} 
    \hline\hline 
    \begin{minipage}[c]{1.0\linewidth}    
      \begin{align*} 
         D^b \left( 
                    \mathbf{WCP}^{m-1|n} 
                    _{ (Q_1, \dots, Q_m | q_1, \ldots , q_n )[s] } 
             \right)^d  
      \end{align*}
      \vskip 11pt 
    \end{minipage} 
    & 
    \begin{minipage}[c]{1.0\linewidth} 
      \begin{equation*}
      \left\langle       
              D^b \left( 
                         \mathbf{WCP}^{m-1}_{(Q_1, \ldots, Q_m)[s, q_1, \dots, q_n]} 
                  \right),\ldots
      \right\rangle      
      \end{equation*}
\vskip 11pt     
    \end{minipage} 
    \\\hline 
    \begin{minipage}[c]{1.0\linewidth} 
      \begin{equation*} 
         \FS\left(\LG 
                      \left( 
                             \mathbf{WCP}^{m-1|n} 
                             _{ (Q_1, \ldots, Q_m | q_1, \ldots, q_n)[s] } 
                      \right)
            \right)^d          
      \end{equation*}
      \vskip 11pt 
    \end{minipage} 
    & 
    \begin{minipage}[c]{1.0\linewidth} 
      \begin{equation*}
      \left\langle 
                   \FS 
                       \left(\LG  
                                 \left( 
                                        \mathbf{WCP}^{m-1} 
                                        _{ (Q_1, \ldots, Q_m)[s, q_1, \dots, q_n] } 
                                 \right) 
                       \right),\ldots
      \right\rangle 
      \end{equation*}
\vskip 11pt     
    \end{minipage} 
    \\\hline  
   \hline\hline 
  \end{tabular} 
  \caption{Super Homological Mirror Symmetry} 
  \label{fig:table_b} 
\end{figure} 
\begin{rem}
Theorems \ref{decomposition} and \ref{decompositionA-side} suggest that different Landau-Ginzburg models can be associated to different noncommutative Hodge substructures in the Hochshild homology of $ D^b $ or Fukaya categories.
\end{rem}

\subsection{Algebro-geometric applications} 
 
The above extension of classical Homological Mirror Symmetry to the case of supermanifolds  
and DG schemes suggests some applications which we discuss next. 
 
\begin{defi}[Fano CY]  
We call a manifold a \emph{Fano CY manifold} if its Hodge diamond contains the  
Hodge structure of a CY (possibly noncommutative) manifold; see Figure~\ref{fig:fanoCY}. 
\end{defi} \pagebreak
 
\begin{figure}[htbp] 
  \centering 
  \includegraphics{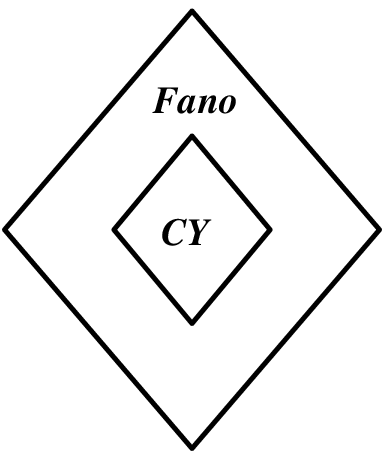} 
  \caption{Definition of Fano CY manifold} 
  \label{fig:fanoCY} 
\end{figure}
 
\begin{example}
\label{7-DFanoCY} 
$ X = \{ \phi_1^3 + \cdots + \phi_9^3 = 0 \} \subset \mathbf{CP}^8 $ is a seven-dimensional Fano CY. 
\end{example}
 
\begin{theorem} 
Let $ X $ be the Fano CY defined in Example \ref{7-DFanoCY}. 
Then $ Gr_3(X) $ is infinitely generated.  
\end{theorem} 
 
\begin{proof} 
It follows from Theorem \ref{decomposition} that  
\begin{equation*} 
D^b(X + \xi_1 \xi_2 + \xi_3 \xi_4) =  D^b(NCY) \, .   
\end{equation*} 
Here, $ NCY $ is a noncommutative Calabi-Yau of superdimension three with a DG scheme structure coming from the anticommuting variables  $ \xi_1, \xi_2 , \xi_3 ,\xi_4 $.
Thus,
\begin{equation*} 
Gr_3(X) = Gr_1(NCY) \, .  
\end{equation*} 
It follows from  \cite{GRIFAT} that we have a super Calabi-Yau integrable system arising from the universal variation of Hodge structures of the noncommutative Calabi-Yau manifold.
A noncommutative modification of a theorem of Voisin \cite{VOI} applied to the above integrable system yields the result that $ Gr_1(NCY) $ is infinitely generated. 
\end{proof} 
The above argument applies to many other examples. 
For more details see \cite{GRIFAT}.

\subsection{Gaps and spectra} 
 
In this subsection, we review the notions of spectra and gaps following \cite{BFK}. 
 
Noncommutative Hodge structures were introduced in \cite{KKP1} as a means of bringing the techniques and tools of Hodge theory into the categorical and noncommutative realm.   
In the classical setting, much of the information about an isolated singularity is contained in the Hodge spectrum (a set of rational eigenvalues of the monodromy operator).   
A categorical analogue of the Hodge spectrum appears in the works of Orlov \cite{Orlov:Remarks} and Rouqier \cite{Rouquier:Dimensions}. 
Let us call this analogue the Orlov spectrum (a rigorous definition will appear below).   
Very little is known about the Orlov spectrum.   
However, recent work \cite{BFK} suggests an intimate connection with the classical theory. 
 
Let us recall the definition of the Orlov spectrum and discuss some of the main results in 
\cite{BFK}.  
Let $ \mathcal T $ be a triangulated category.   
For any $ G \in \mathcal T $, denote by $ \langle G \rangle_0 $ the smallest full subcategory containing $ G $ which is closed under isomorphisms, shifting, and taking finite direct sums and summands.  
Now, inductively define $ \langle G \rangle_n $ as the full subcategory of objects $ B $ such that there is a distinguished triangle  
$ X \to B \to Y \to X $, with $ X \in \langle G \rangle_{n-1} $ and  
$ Y \in \langle G \rangle_0 $. 
 
\begin{defi}
\label{OrlovSpec} 
Let $ G $ be an object of a triangulated category $ \mathcal{T} $.  
If there is an $ n $ with $ \langle G \rangle_{n} = \mathcal T $, we set  
\begin{displaymath} 
 \tritime(G):=  \text{min } \lbrace  n \geq 0 \  | \ \langle G 
 \rangle_{n} = \mathcal T \rbrace.  
\end{displaymath} 
Otherwise, we set $\tritime(G) := \infty$.  
We call $ \tritime(G) $ the \textbf{generation time} of $ G $. 
If $ \tritime(G) $ is finite, we say that $ G $ is a \textbf{strong generator}. 
The \textbf{Orlov spectrum} of $ \mathcal T $ is the union of all possible generation times for 
strong generators of $ \mathcal T $.  
The \text{Rouquier dimension} is the smallest number in the Orlov spectrum.  
We say that a triangulated category $ \mathcal T $ has a \textbf{gap} of length $ s $, if $ a $ and 
$ a+s+1 $ are in the Orlov spectrum but $ r $ is not in the Orlov spectrum for $ a < r < a+s+1 $. 
\end{defi} 
 
\begin{conj} 
\label{gap bound} 
If $ X $ is a smooth variety then any gap of $ D^{b}(X) $ has length at most 
the Krull dimension of $ X $. 
\end{conj} 
In the noncommutative situation, the gaps could be larger.  
We will now explain how the notions of gaps and spectra connect with supercohomology calculations.

\subsection{Supermanifolds and exotic \texorpdfstring{$ (p,p) $}{(p,p)} cycles} 
 
In this section, we look at Example \ref{Counterexample} and Example \ref{ExHyper2} from the  
perspective of the Hodge conjecture, i.e. every $ (p,p) $ cycle is algebraic. 
We begin with the following. 
 
\begin{conj}
\label{SuperMatrixFactorization} 
(``super matrix factorization'' version of the Hodge conjecture; see \cite{FK}) 
Every $ (p,p) $ class in the Jacobian ideal (Hochschild cohomology) acts nontrivially on the category of singularities. 
\end{conj} 
 
\begin{conj}
\label{SuperHodgeGapCorrespondence} 
There is a correspondence between ``missing'' supermanifold Hodge numbers and increasing gaps of categories of singularities. 
\end{conj} 
The supermanifolds of Examples \ref{Counterexample} and \ref{ExHyper2} are candidates for displaying the behavior described in Conjecture \ref{SuperHodgeGapCorrespondence}.
Other candidates may be obtained by following the procedure described in \cite{Sethi,GaravusoKreuzerNoll:Fano}. 
 
Conjectures \ref{SuperMatrixFactorization} and \ref{SuperHodgeGapCorrespondence} are based on the (not yet completely established) equality 
\begin{equation}
\label{tritime(G)} 
\tritime(G) = \tritime(G+G') + B(G') \cdot l_{G'}(G) - ME \, .
\end{equation} 
Here, $ l_{G'}(G) $ is the number of steps in which $ G' $ generates $ G $, $ B(G') $ measures how far $ \textrm{End}(G') $ is from being formal, and $ ME $ is the monodromy effect on the mirror side. 
 
\begin{rem} 
Theorem \ref{decomposition} and Equation (\ref{tritime(G)}) explains the connection with  
the supercohomology calculations indicated in Examples \ref{Counterexample} and \ref{ExHyper2}. 
Indeed what needs to be computed there are the Hochshild cohomologies of deformed categories; gaps are responsible for ``missing'' Hodge numbers. 
In order to do the calculations, one needs to introduce an enhanced notion of spectra.
This will be explained in detail in \cite{GaravusoKatzarkovNoll}. 
\end{rem}

Based on Conjectures \ref{SuperMatrixFactorization} and \ref{SuperHodgeGapCorrespondence}, two kinds of spectral anomalies are possible:   
\begin{enumerate}

\item 
Big gap drops as a consequence of appearance of new $ (p,p) $ classes. 
The first two rows of Table \ref{tab:HC31} and the tropical Abelian varieties of \cite{KRIS} are examples.  

\item 
Preservation of big gaps, i.e. rather small disappearance of gaps as a consequence of new $ (p,p) $ classes.
The third row of Table \ref{tab:HC31} is an example. 

\end{enumerate} 
  
\begin{table}
\begin{center} 
\begin{tabular}{|c|c|} 
\hline 
\begin{minipage}[c]{2.75in} 
\centering 
\medskip

Supermanifolds
 
$ M = \{ \phi_1^3 + \cdots + \phi_9^3 + \xi_1 \xi_2 + \xi_3 \xi_4 = 0 \} $
$ \in \mathbf{WCP}^{8|4}_{(1,1,1,1,1,1,1,1,1|q_1,q_2,q_3,q_4)[3]} $ 
\medskip 
 
\end{minipage} 
& 
\begin{minipage}[c]{2.25in} 
\centering 
\medskip

$ X = \{ \phi_1^3 + \cdots + \phi_9^3  = 0 \} \subset \mathbf{CP}^8 $
\vskip 11pt
$ Gap( D^b(X) ) = 6 $
 
$ Gap( D^b(M) ) = 4 $ 
 
\medskip 
 
\end{minipage} 
\\ 
\hline 
\begin{minipage}[c]{2.75in} 
\centering 
\medskip 
 
Two-dimensional noncommutative tori $ T $ 
 
$$\includegraphics[width = 0.4\textwidth]{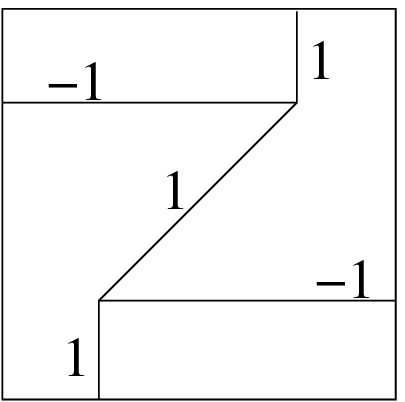}$$ 
 
\medskip 
 
\end{minipage} 
& 
\begin{minipage}[c]{2.25in} 
\centering 
\medskip 
 
$ Gap(\Fuk(Generic)) = \infty $ 
 
$ Gap(\Fuk(T)) = 3 $ 
 
\end{minipage} 
\\\hline 
\begin{minipage}[c]{2.75in} 
\centering 
\medskip 
 
Voisin--Thomas 
 
$$\includegraphics[width = 0.7\textwidth]{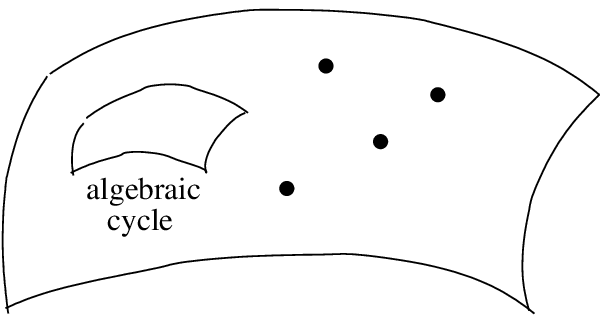}$$ 
 
Too few Lagrangian spheres in superhypersurfaces 
 
\medskip 
 
\end{minipage} 
& 
\begin{minipage}[c]{2.25in} 
\centering 
\medskip 
 
$ Gap(Generic)=n $ 
 
$ Gap(Special)=n-k $,  $ k <  2 $ 
 
\medskip 
 
\end{minipage} 
\\\hline 
\end{tabular} 
\end{center} 
\caption{Gap drops} 
\label{tab:HC31} 
\end{table}
 
In Table \ref{tab:HC32}, we connect small gap drops with coniveau filtration; see \cite{BEI}.
The first row of the table explains the action of elements of the Jacobian ring and tropical
$ (p,p) $ cycles on the category of matrix factorizations and Fukaya categories, respectively. 
The trace formula shows when this action is trivial.
This allows us to compute $ l_{G'}(G) $ and $ B(G') $. 
Based on these calculations, in the second row of the table, we connect Beilinson coniveu filtration $ N'(H^*) $ with relatively small gap drops. 
  
\begin{table}
\begin{center} 
\begin{tabular}{|c|c|} 
\hline 
\begin{minipage}[c]{2.75in} 
\centering 
\medskip 
 
Jacobian ideal action in super, tropical and standard cases. 
 
\begin{itemize} 
  \item $\alpha\in J_f$ 
  \item $\alpha\in J_f$ super 
  \item $\alpha\in H_{trop}$ 
  \item Defects in $J_f$ determine gaps 
\end{itemize} 
 
$$\includegraphics[width = 0.7\textwidth]{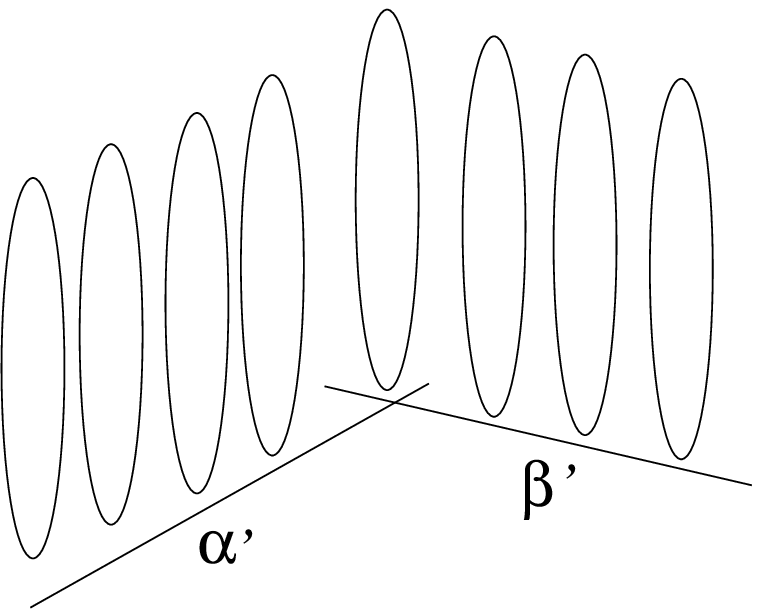}$$ 
 
\medskip 
 
\end{minipage} 
& 
\begin{minipage}[c]{2.75in} 
\centering 
\medskip 
 
$ A\int_\Delta \frac{tr(\alpha,\beta)}{\partial_1w\ldots\partial_nw}=0 $ 
 
\medskip 
 
\end{minipage} 
\\ 
\hline 
\begin{minipage}[c]{2.75in} 
\centering 
\medskip 
 
Beilinson coniveu filtration  
 
$ \alpha \in N'(H^*) $, 
 
$ \alpha $ is the slowest  generator 
 
$ \alpha' $ is a generator connected  
 
with a new $ (p,p) $ class  
 
\medskip 
 
\end{minipage} 
& 
\begin{minipage}[c]{2.75in} 
\centering 
\medskip 
 
The gap for $ \alpha $ is $ n $. 
 
The gap with new class $ \alpha' $ is $ n-1 $ 
 
\medskip 
 
\end{minipage} 
\\\hline 
\end{tabular} 
\end{center} 
\caption{Small gap drops for coniveau elements.} 
\label{tab:HC32} 
\end{table} 
 
Certain strictness properties \cite{FK} may obstruct the appearance of small gap drops and hence yield counterexamples of the Hodge conjecture.  
As will be discussed in \cite{GaravusoKatzarkovNoll}, we may formulate the following conjecture.   
\begin{conj}  
The dimension of supercohomology is determined by Jacobian ring, gaps, and monodromy effect $ ME $.  
\end{conj} 
This conjecture implies that, in Example \ref{Counterexample}, 
$ h^{(1,2)}\left( M^{(q_1,q_2)} \right) = 83 $  when $ (q_1, q_2) \in \{ (1,5),(2,4) \} $. 
Thus, there are possible anomalous cycles.
For a detailed explanation, see \cite{GaravusoKatzarkovNoll,FK}. 
\vfill 
 
\section*{Acknowledgements} 
 
The authors thank V. Bouchard, C. Doran, M. Kontsevich, T. Pantev, and C. Vafa 
for useful discussions. 
R.G. and M.K. received financial support from the Austrian Research Funds (FWF) 
under grant numbers I192 and P21239. 
R.G. received additional financial support from NSERC.  
L.K. was funded by NSF Grant DMS0600800, NSF FRG Grant DMS-0652633, FWF  
Grant P20778, and an ERC Grant. 
A.N. received financial support from FWF Grant P20778. 
 
\begin{appendix} 
 
\section{A-model super CY hypersurface  
\texorpdfstring{$ \Leftrightarrow $}{<-->} B-model quasihomogeneity} 
 
In this appendix, we will prove a theorem concerning the  
quasihomogeneity of $ \widetilde{W} $ given by (\ref{Wtilde}) and  
prove a corollary concerning the $ t \rightarrow - \infty $ limit. 
 
\begin{theorem} 
\label{SCY-QH}  
A member of the hypersurface family 
$$ \mathbf{WCP}^{m-1|n}_{ (Q_1,\ldots,Q_m |q_1,\ldots,q_n)[s] } $$  
is super Calabi-Yau  
if and only if  
$$  
\widetilde{W} = \sum_{i=1}^m  \prod_{j=1}^m y_j^{ M_{ji} } 
    + e^{t/s} \prod_{j=1}^m y_j \prod_{b=1}^n x_b^{-1} 
    + \sum_{a=1}^n 
       \left( 
             1 + x_a^{-2} \hat{\eta}_a \hat{\gamma}_a 
      \right) 
      \prod_{b=1}^n x_b^{ N_{ba} } $$  
is quasihomogeneous of some degree $ s^{\prime} $ for all values of $ t $.  
\end{theorem} 
\begin{proof}  
Assume that $ \widetilde{W} $ is quasihomogeneous of degree $ s^{\prime} $ for all values of $ t $. 
It follows that 
\begin{equation} 
\label{QHs} 
s^{\prime} 
   = \sum_{j=1}^m n_{y_j} M_{ji}  
   = \sum_{b=1}^n n_{x_b} N_{ba} 
   = \sum_{j=1}^m n_{y_j} - \sum_{b=1}^n n_{x_b} 
\end{equation}  
and 
\begin{equation} 
\label{QH0} 
0 = -2 n_{x_a} + n_{\hat{\eta}_a} + n_{\hat{\gamma}_a} \, ,  
\end{equation} 
where  
$ n_{y_j} $, $ n_{x_b} $,  
$ n_{\hat{\eta}_a} $, and $ n_{\hat{\gamma}_a} $ 
are the weights of  
$ y_j $, $ x_b $,  
$ \hat{\eta}_a $, and $ \hat{\gamma}_a $, respectively.   
Combining the first two equalities of (\ref{QHs}) with  
(\ref{MandNequations})  
yields 
\begin{equation} 
\label{A-BWeightRelations} 
\sum_{i=1}^m n_{y_i} 
   = \frac{s^{\prime}}{s} \sum_{i=1}^m Q_i \, , 
\quad 
\sum_{a=1}^n n_{x_a} 
   = \frac{s^{\prime}}{s} \sum_{a=1}^n q_a \, . 
\end{equation} 
Using (\ref{A-BWeightRelations}) in (\ref{QHs}) yields 
$$ 
\sum_{i=1}^m Q_i - \sum_{a=1}^n q_a = s \, , 
$$ 
which is the super Calabi-Yau condition. \pagebreak
 
To prove the converse, assume that the super Calabi-Yau condition holds. 
It is always possible to find $ n_{y_j} $ and $ n_{x_b} $  
which satisfy the first two equalities of (\ref{QHs}).   
Furthermore, it is always possible to satisfy (\ref{QH0}).     
Thus, the first and third terms of $ \widetilde{W} $ can always be  
chosen to be quasihomogeneous of some degree $ s^{\prime} $. 
It remains to show that the second term of $ \widetilde{W} $ is  
quasihomogeneous of degree $ s^{\prime} $ for all values of $ t $.   
Combining the first two equalities of (\ref{QHs}) with  
(\ref{MandNequations}) once again yields (\ref{A-BWeightRelations}). 
Using (\ref{A-BWeightRelations}) in the super Calabi-Yau condition yields 
$$ 
\sum_{j=1}^m n_{y_j} - \sum_{b=1}^n n_{x_b} = s^{\prime} 
$$ 
and hence the second term of $ \widetilde{W} $ is quasihomogeneous of  
degree $ s^{\prime} $ for all values of $ t $. 
\end{proof} 
\begin{remark} 
A similiar phenomenon was observed  
\cite{BelhajDrissiRasmussenSaidiSebbar:Toric}  
when the A-model target space is a crepant resolution of 
$$ 
\mathbf{WCP}^{m_1-1|n}_{  
                         ( Q_{11},\ldots,Q_{1 m_1} | q_{11},\ldots, q_{1n} ) 
                       } 
\times \cdots \times 
\mathbf{WCP}^{m_n-1|n}_{ 
                         ( Q_{n1},\ldots,Q_{n m_n} | q_{n1},\ldots, q_{nn} ) 
                       } \, . 
$$ 
\end{remark}
 
\begin{corollary}   
Let $ \widetilde{W} $ be the superpotential of the  
super Landau-Ginzburg orbifold which is mirror to a member  
$ \{ G= 0 \}  $ of the hypersurface family 
$ \mathbf{WCP}^{m-1|n}_{ (Q_1,\ldots,Q_m |q_1,\ldots,q_n)[s] } $.  
Then, in the limit $ t \rightarrow - \infty $, the following are true: 
\newcounter{Lcount} 
\begin{list}{(\roman{Lcount})}{} 
 
\usecounter{Lcount} 
 
\item $ \widetilde{W} $ is quasihomogeneous of some degree $ s^{\prime} $. 
 
\item 
The quantity  
\begin{equation*} 
\hat{c}  
  =   \sum_{i=1}^m 
      \left( 
             1 - \frac{2 Q_i}{s} 
      \right)  
    - \sum_{a=1}^n  
      \left( 
             1 - \frac{2 q_a}{s} 
      \right)  
\end{equation*} 
associated with $ G $ is equal to the quantity 
\begin{equation*} 
\hat{\tilde{c}}  
  =   \sum_{i=1}^m 
      \left( 
             1 - \frac{ 2 n_{y_i} }{s^{\prime}} 
      \right) 
    + \sum_{a=1}^m 
      \left( 
             1 - \frac{ 2 n_{x_a} }{s^{\prime}} 
      \right) 
    - \sum_{a=1}^n 
      \left( 
             1 - \frac{ 2 n_{\hat{\eta}_a} }{s^{\prime}} 
      \right)     
    -  \sum_{a=1}^n 
      \left( 
             1 - \frac{ 2 n_{\hat{\gamma}_a} }{s^{\prime}} 
      \right)  
\end{equation*} 
associated with $ \widetilde{W} $. 
 
\end{list} 
\end{corollary}
 
\begin{proof} 
\newcounter{LLcount} 
\begin{list}{(\roman{LLcount})}{} 
 
\usecounter{LLcount} 
 
\item[]
 
\item This follows from the discussion in the second paragraph of the  
proof of Theorem \ref{SCY-QH}.
 
\item This follows by combining the terms in $ \hat{\tilde{c}} $ which are  
summed over $ a $ and then using (\ref{QH0}) and  
(\ref{A-BWeightRelations}) in turn, i.e. \pagebreak
\begin{align*} 
{\hat{\tilde{c}}}  
  &= \sum_{i=1}^m 
     \left( 
            1 - \frac{ 2 n_{y_i} }{s^{\prime}} 
     \right) 
   - \sum_{a=1}^m 
     \left[ 
             1  
           - \frac{2}{s^{\prime}}  
             \left(  
                    n_{\hat{\eta}_a} + n_{\hat{\eta}_a} - n_{x_a}  
            \right)  
     \right]     
\\
  &= \sum_{i=1}^m 
     \left( 
            1 - \frac{ 2 n_{y_i} }{s^{\prime}} 
     \right) 
   - \sum_{a=1}^n 
     \left( 
            1 - \frac{ 2 n_{x_a} }{s^{\prime}} 
     \right) 
\\
  &= \sum_{i=1}^m 
      \left( 
             1 - \frac{2 Q_i}{s} 
      \right) 
    - \sum_{a=1}^n 
      \left( 
             1 - \frac{2 q_a}{s} 
      \right)  
   = \hat{c} \, . \qedhere
\end{align*}   
 
\end{list} 
 
\end{proof}

\end{appendix}


\end{document}